\documentclass[acmsmall]{acmart}
\renewcommand{\paragraph}[1]{\medskip\noindent\textbf{\emph{#1.}}}

\acmJournal{PACMPL}
\acmVolume{1}
\acmNumber{CONF} % CONF = POPL or ICFP or OOPSLA
\acmArticle{1}
\acmYear{2018}
\acmMonth{1}
\acmDOI{} % \acmDOI{10.1145/nnnnnnn.nnnnnnn}
\startPage{1}

\setcopyright{none}
\usepackage{booktabs}   %% For formal tables:
\usepackage{subcaption} %% For complex figures with subfigures/subcaptions
\usepackage{graphicx}
\usepackage{amscd,amsmath}
\usepackage{verbatim}
\usepackage[all]{xy}
\usepackage{fancyvrb}
\usepackage{xcolor}
\usepackage{color}
\usepackage{diagbox}

\usepackage{enumitem}

\usepackage{booktabs}

\usepackage{microtype}

\usepackage{graphicx}        % standard LaTeX graphics tool

\usepackage{array}

\usepackage{multicol}

\usepackage{listings}
\definecolor{codegreen}{rgb}{0,0.6,0}
\definecolor{codegray}{rgb}{0.5,0.5,0.5}
\definecolor{codepurple}{rgb}{0.58,0,0.82}
\definecolor{backcolour}{rgb}{0.95,0.95,0.92}

\usepackage{dashbox}
\newcommand{\codetemplate}[1]{\dbox{$#1$}}
\newcommand{\Hole}[1]{\dbox{$#1$}}
\newcommand{\errorRejected}{\textit{Failed}}
\newcommand{\errorRose}{\textit{Failed}}
\newcommand{\errorUnsat}{\textit{Unsat}}
\newcommand{\errorTimeout}{\textit{Timed out}}
\usepackage{paralist}
\usepackage{lscape}

\renewcommand{\baselinestretch}{0.96}
\usepackage{rotating}

\definecolor{gruen}{rgb}{0.0, 0.5, 0.0}
\definecolor{rot}{rgb}{1.0, 0.13, 0.32}
\newcommand{\success}{\textcolor{gruen}{\textbf{Success}}}
\newcommand{\fail}{\textcolor{rot}{\textbf{Failed}}}
\newcommand{\unsat}{\textcolor{rot}{\textbf{Unsat}}}
\newcommand{\tl}{\textcolor{rot}{\textbf{Timed out}}}
\newcommand{\wa}{\textcolor{rot}{\textbf{Wrong}}}

\begin{document}

%% Title information
\title[Synthesis of Polynomial Programs]{Template-based Program Synthesis using Stellensätze}         %% [Short Title] is optional;
                                        %% when present, will be used in
                                        %% header instead of Full Title.
%\titlenote{with title note}             %% \titlenote is optional;
                                        %% can be repeated if necessary;
                                        %% contents suppressed with 'anonymous'
%\subtitle{Subtitle}                     %% \subtitle is optional
%\subtitlenote{with subtitle note}       %% \subtitlenote is optional;
                                        %% can be repeated if necessary;
                                        %% contents suppressed with 'anonymous'

\author{Amir Kafshdar Goharshady}
\authornote{Authors are ordered alphabetically.}          %% \authornote is optional;
                                        %% can be repeated if necessary
\orcid{0000-0003-1702-6584}             %% \orcid is optional
\affiliation{
  \department{Department of Computer Science and Engineering}
  \department{Department of Mathematics}
  \institution{Hong Kong University of Science and Technology}            %% \institution is required
  \city{Clear Water Bay, Kowloon}
  \country{Hong Kong}                    %% \country is recommended
}
\email{goharshady@cse.ust.hk}          %% \email is recommended

\author{S. Hitarth}
%\authornote{with author1 note}          %% \authornote is optional;
%% can be repeated if necessary
\orcid{0000-0001-7419-3560}             %% \orcid is optional
\affiliation{
	\department{Department of Computer Science and Engineering}
	\institution{Hong Kong University of Science and Technology}            %% \institution is required
	\city{Clear Water Bay, Kowloon}
	\country{Hong Kong}                    %% \country is recommended
}
\email{hsinghab@connect.ust.hk}          %% \email is recommended

\author{Fatemeh Mohammadi}
%\authornote{with author1 note}          %% \authornote is optional;
%% can be repeated if necessary
\orcid{0000-0001-5187-0995}             %% \orcid is optional
\affiliation{
	\department{Department of Mathematics and Department of Computer Science, KU Leuven, Belgium}
	\institution{KU Leuven}            %% \institution is
	\city{Leuven}
	\country{Belgium}                    %% \country is recommended
}
\email{fatemeh.mohammadi@kuleuven.be}          %% \email is recommended

\author{Harshit Jitendra Motwani}
%\authornote{with author1 note}          %% \authornote is optional;
%% can be repeated if necessary
\orcid{0000-0002-2142-4254}             %% \orcid is optional
\affiliation{
	\department{Department of Mathematics: Algebra and Geometry}
	\institution{Ghent University}            %% \institution is required
	\city{Ghent}
	\country{Belgium}                    %% \country is recommended
}
\email{harshitjitendra.motwani@ugent.be}          %% \email is recommended

%% Abstract
%% Note: \begin{abstract}...\end{abstract} environment must come
%% before \maketitle command
\begin{abstract}

Template-based synthesis, also known as sketching, is a localized approach to program synthesis in which the programmer provides not only a specification, but also a high-level ``sketch'' of the program. The sketch is basically a partial program that models the general intuition of the programmer, while leaving the low-level details as unimplemented ``holes''. The role of the synthesis engine is then to fill in these holes such that the completed program satisfies the desired specification. In this work, we focus on template-based synthesis of polynomial imperative programs with real variables, i.e.~imperative programs in which all expressions appearing in assignments, conditions and guards are polynomials over program variables. While this problem can be solved in a sound and complete manner by a reduction to the first-order theory of the reals, the resulting formulas will contain a quantifier alternation and are extremely hard for modern SMT solvers, even when considering toy programs with a handful of lines. Moreover, the classical algorithms for quantifier elimination are notoriously unscalable and not at all applicable to this use-case.

In contrast, our main contribution is an algorithm, based on several well-known theorems in polyhedral and real algebraic geometry, namely Putinar's Positivstellensatz, the Real Nullstellensatz, Handelman's Theorem and Farkas' Lemma, which sidesteps the quantifier elimination difficulty and reduces the problem directly to Quadratic Programming (QP). Alternatively, one can view our algorithm as an efficient way of eliminating quantifiers in the particular formulas that appear in the synthesis problem. The resulting QP instances can then be handled quite easily by SMT solvers. Notably, our reduction to QP is sound and semi-complete, i.e.~it is complete if polynomials of a sufficiently high degree are used in the templates. Thus, we provide the first method for sketching-based synthesis of polynomial programs that does not sacrifice completeness, while being scalable enough to handle meaningful programs. Finally, we provide experimental results over a variety of examples from the literature.%, including programs modeling physical phenomena, illustrating the applicability of our approach.  

%Program synthesis is a \textit{search problem} where the task is to generate a program that meets a given specification. The specification is usually given in the form of input/output pairs (Programming By Example), or logical specifications (Constraint Solving). There exist approaches based on user-provided insights in the form of \textit{templates} that prune the search space of solutions. The approach we present generates a non-linear constraint system whose solution would give us the synthesised program. We rely on SMT solvers to solve the resulting constraint systems. For linear systems, the existing method relies on Farkas' Lemma to eliminate quantifier alternation and use off-the-shelf SMT solvers for synthesis. In this paper, we extend this approach to polynomial system where the \textit{holes} are to be replaced by appropriate polynomial expressions of specified degree in the program variables predefined by the user in the template. We provide sound and semi-complete synthesis algorithm for these programs which uses Positivstellensatz, Real Nullstellensatz and Handelman's theorem from Real Algebraic Geometry to eliminate quantifier alternation and then finally use SMT solvers.
\end{abstract}


%% 2012 ACM Computing Classification System (CSS) concepts
%% Generate at 'http://dl.acm.org/ccs/ccs.cfm'.
%\begin{CCSXML}
%<ccs2012>
%<concept>
%<concept_id>10011007.10011006.10011008</concept_id>
%<concept_desc>Software and its engineering~General programming languages</concept_desc>
%<concept_significance>500</concept_significance>
%</concept>
%<concept>
%<concept_id>10003456.10003457.10003521.10003525</concept_id>
%<concept_desc>Social and professional topics~History of programming languages</concept_desc>
%<concept_significance>300</concept_significance>
%</concept>
%</ccs2012>
%\end{CCSXML}
%
%\ccsdesc[500]{Software and its engineering~General programming languages}
%\ccsdesc[300]{Social and professional topics~History of programming languages}
%% End of generated code

%% Keywords
%% comma separated list
 %% \keywords are mandatory in final camera-ready submission

%% \maketitle
%% Note: \maketitle command must come after title commands, author
%% commands, abstract environment, Computing Classification System
%% environment and commands, and keywords command.
\maketitle

\section{Introduction} \label{sec:intro}

\paragraph{Program Synthesis} Program synthesis is one of the most central research topics not only in the programming languages and verification communities, but in computer science as a whole~\cite{pnueli1989synthesis}. It is often called the holy grail of computing~\cite{gulwani2017program}. Generally speaking, the goal in program synthesis is to automate the process of programming, i.e.~given a specification, which is often a logical formula, automatically find a program (in a predefined language) that satisfies the specification. 

\paragraph{Short History} Various formulations of this problem have been studied since the 1950s~(e.g.~see \cite{church1963application} which was originally published in 1957). Similar problems were also considered in the context of constructive mathematics as far back as the 1930s~\cite{troelstra1977aspects,kolmogoroff1932deutung}. Today, there is a rich body of literature that uses widely different techniques for synthesis. Some notable approaches to program synthesis include deductive synthesis using theorem provers~\cite{green1981application,manna1971toward}, counterexample-guided synthesis~\cite{preiner2017counterexample,reynolds2015counterexample,abate2018counterexample}, evolutionary algorithms~\cite{koza1994genetic,krawiec2016behavioral,sobania2022comprehensive},  sketching/syntax-guided~\cite{alur2013syntax,solar2008program,solar2009sketching,si2018syntax,fedyukovich2019quantified}, template-based/search-based~\cite{alur2018search,srivastava2013template}, proof-theoretic methods~\cite{srivastava2010program}, constraint-solving~\cite{polikarpova2016program,gulwani2011synthesis}, type-guided~\cite{guo2019program,polikarpova2016program}, resource-guided~\cite{knoth2019resource}, and synthesis based on examples (input-output pairs)~\cite{gulwani2012spreadsheet,shaw1975inferring,smith1975pygmalion,gulwani2011automating}. This is by no means a comprehensive list and there is no way we can do justice to all the work in this area. Moreover, most approaches, including ours, utilize more than one of the paradigms named above. Therefore, we refer to~\cite{gulwani2017program} for a more detailed discussion.

\paragraph{Search-based/Syntax-guided Synthesis} While the initial classical formulations of program synthesis focused on automated generation of the whole program, given only a specification as a logical/semantic-based formula, it was later realized that scaling this approach is quite challenging and much efficiency can be gained by allowing the programmer to also provide syntactic templates that limit the search space of possible implementations~\cite{alur2018search}. Intuitively speaking, this would let the programmer focus on the high-level ideas behind the program, while relying on the synthesis engine to fill in the low-level details.

Various formalizations of similar ideas have appeared in the literature. For example, in~\cite{alur2013syntax}, the programmer has to specify a background theory, a correctness specification, which is semantic in nature, and a grammar which is used to syntactically limit the search space of possible implementations. In \emph{sketching}~\cite{solar2008program}, the programmer provides the high-level structure or ``sketch'' of the implementation, while leaving low-level ``holes'' for the synthesis engine to fill. In template-based approaches, such as~\cite{srivastava2013template}, the programmer additionally provides a template for each hole, e.g.~specifying that the hole should be filled by a linear/affine expression or that it should likely utilize a specific set of variables. These approaches have all led to highly successful synthesis engines. In particular, the SyGuS competition~\cite{syguscomp} is a great showcase of the remarkable year-on-year progress in program synthesis.  

\paragraph{Completeness} Given that most variants of program synthesis are undecidable, often due to Rice's theorem~\cite{rice1953classes}, it is no surprise that the synthesis algorithms cannot be sound, complete and terminating at the same time. The undecidability also applies to almost all cases of syntax-guided synthesis, since they are more expressive than the verification of specific non-trivial semantic properties. As such, most synthesis engines, including the ones mentioned above, choose to sacrifice completeness and focus on efficient sound algorithms that can handle a wide variety of real-world use-cases. Completeness can only be achieved in limited circumstances, such as template-based synthesis of linear/affine programs with linear/affine invariants~\cite{srivastava2013template}. So, a natural question is to explore and characterize families of programs for which we can find efficient, sound and complete synthesis procedures. Of course, the difficulty lies in obtaining all three desired properties at the same time and not compromising on any of them.

\paragraph{Our Setting} In this work, we consider template-based synthesis over polynomial programs, i.e.~imperative programs with bounded real variables in which all expressions appearing in the assignments, loop guards, conditionals, and invariants (pre and post-conditions) are polynomials over program variables. See Section~\ref{sec:problem-def} for a more formal statement. Following the sketching/template-based approach, we assume that the programmer provides a sketch of the program, together with one or more holes. The holes can appear both in the program itself or in the invariants and will be filled by polynomial expressions. The programmer can also provide a template for how each hole should be filled by specifying the program variables that can appear in the hole and the expected degree of the polynomial expression used for filling it. Note that this is without any loss of generality, since, in the absence of such information, our algorithm can simply iterate over all possible templates with a bounded degree.

\begin{figure}
	\begin{center}
		\begin{minipage}{0.6\linewidth}
		\begin{lstlisting}[language=C,mathescape=true]
@real: $i$, $s$, $n$; 
@pre: $n \geq 1$;
$i = 0$;
$s = 0$;
@invariant: $\codetemplate{[\{i, s\},2]} \geq 0~\wedge~ \codetemplate{[\{i, s\},2]} \geq 0 ~\wedge~ \codetemplate{[\{i, n\}, 1]} \geq 0$
while($~\codetemplate{[\{i, n\}, 1]}~$)
{
	$i$ = $i+1$;
	$s$ = $\codetemplate{[\{i, s\},1]}$;
}   
@post: $(n-1) \cdot n / 2 \leq s \leq n \cdot (n+1) / 2$;
		\end{lstlisting}
		\end{minipage}
	\end{center}
	\vspace{-1.5em}
	\caption{An example polynomial program with holes.}
	\label{fig:progintro}
\end{figure}

\begin{example}

	As a simple example, consider the program in Figure~\ref{fig:progintro}, in which we have three bounded real variables $i, s,$ and $n$. The programmer has given us a sketch of the program, which includes the desired specification as a precondition and a postcondition. However, it also includes holes that should be filled by the synthesis engine. These holes are shown by dashed boxes. In this case, the programmer asks us to synthesize an invariant and a guard for the while loop, as well as an expression for the right-hand side of the assignment inside the while loop. Moreover, she has provided us with a template for each hole, specifying the variables that she expects would need to appear in that hole and the maximum degree of the polynomial that should be synthesized for filling the hole. For example, the while guard must be an affine (degree 1) polynomial over the variables $i$ and $n$. As we will see in the sequel, given this partial program as input, our approach is able to synthesize
	the completed program in Figure~\ref{fig:progintrocomplete}, which basically sums up all the integers from $1$ to $\lfloor n \rfloor.$

	\begin{figure}[h]
		\begin{center}
		\begin{minipage}{0.52\linewidth}
		\begin{lstlisting}[language=C,mathescape=true]
@real: $i$, $s$, $n$; 
@pre: $n \geq 1$;
$i = 0$;
$s = 0$;
@invariant: $s \geq i \cdot (i+1)/2 ~\wedge~ s \leq i \cdot (i+1)/2 ~\wedge~ i \leq n$
while($i \leq n - 1$)
{
	$i$ = $i+1$;
	$s$ = $s+i$;
}   
@post: $(n-1) \cdot n / 2 \leq s \leq n \cdot (n+1) / 2$;
		\end{lstlisting}
	\end{minipage}
\end{center}
		\vspace{-1.5em}
		\caption{The synthesized (completed) version of the partial program in Figure~\ref{fig:progintro}.}
		\label{fig:progintrocomplete}
	\end{figure}
	
\end{example}

\paragraph{Decidability and Hardness} Note that the synthesis problem we consider here can be reduced, in polynomial time, to the decision problem of a formula in the first-order theory of the reals and is hence decidable. See Section~\ref{sec:complexity} for more details. However, the resulting formula is quite long and has a quantifier alternation. Thus, it is beyond the reach of modern SMT solvers and quantifier elimination methods, even for toy programs with three lines of code. In the same section, we also show that the problem is (strongly) NP-hard, so no strictly PTIME algorithms can be expected unless P=NP.  The NP-hardness holds even in the special case of linear/affine programs.

\paragraph{Our Contributions} In this work, we consider the problem of template-based synthesis over polynomial imperative programs with bounded real variables. Our contributions are as follows:
\begin{compactitem}
	\item \emph{Complexity:} We prove that the problem is decidable. This is achieved by a reduction to the first-order theory of the reals. We also provide a reduction showing that our problem is (strongly) NP-hard. Note that, as mentioned above, decision procedures for the first-order theory of the reals are notoriously slow and cannot handle even toy programs in practice.
	\item \emph{Practical Algorithm:} We use classical theorems from polyhedral and real algebraic geometry, including Farkas' Lemma, Handelman's Theorem, Putinar's Positivstellensatz\footnote{Positivstellensatz is German for ``positive locus theorem''. Its plural form is Positivstellens\"atze.}, and the Real Nullstellensatz\footnote{Nullstellensatz is German for ``zero locus theorem''.} to obtain a polynomial-time reduction to Quadratic Programming (QP). Since modern solvers, both SMT solvers and numerical ones, are quite efficient in handling real-world QP instances, this leads to a much more scalable approach to program synthesis in comparison to decision procedures for the first-order theory of the reals.
	
	\item \emph{Completeness:} Crucially, we prove that our approach is not only sound but also \emph{semi-complete}, i.e.~it is complete as long as polynomials of sufficiently high degree are used in the templates. So, our approach achieves both soundness and completeness, while not suffering from the unscalability of quantifier elimination and decision procedures for the first-order theory of the reals.
	
	\item \emph{Experimental Results:} We provide extensive experimental results, illustrating the efficiency of our approach and that it can successfully handle a variety of programs. %Note that programs which model/simulate physical phenomena are often polynomial and hence an ideal domain for our approach.
\end{compactitem}

\paragraph{Novelty} The main novelty of our approach is that it can handle a wide family of programs, i.e.~polynomial imperative programs, without sacrificing completeness. It is quite rare in program synthesis to achieve soundness, semi-completeness and practical efficiency at the same time. To the best of our knowledge, our method of relying on theorems in real algebraic geometry is also novel in program synthesis. Moreover, we combine these theorems in non-trivial ways and obtain new mathematical results, which are of independent interest. 

\paragraph{Limitations} The main limitation of our approach is that it can only handle polynomial programs. This is to be expected, since we are using techniques from real algebraic geometry, which are highly dependent on having polynomial expressions/inequalities over real variables. Another limitation is that, although our approach is quite efficient in practice, its runtime is not polynomial in the size of the program or the number of variables. This is a natural consequence of our NP-hardness result and cannot be overcome unless P=NP.   

\paragraph{Related Works} Several previous works use ideas that are similar in spirit to our approach:
\begin{compactitem}
	\item \emph{Template-based Synthesis:} The work~\cite{srivastava2013template} provides algorithms based on Farkas' Lemma for template-based synthesis of linear/affine programs. Since our work can handle polynomial programs of arbitrary degree, \cite{srivastava2013template}'s results on linear/affine programs are subsumed by our approach. Moreover, while the use of Farkas' Lemma is shared among the two approaches, we also use theorems from real algebraic geometry, e.g. Handelman and Putinar, as well as non-trivial combinations of them, which are not considered in~\cite{srivastava2013template}. So, we have a different methodology and handle a wider family of programs. To the best of our knowledge, our algebro-geometric approach is novel in program synthesis.
	
	\item  \emph{Invariant Generation:} The work~\cite{pldi2020} provides template-based algorithms for automated generation of polynomial invariants over polynomial programs with real variables. Note that this is also a special case of our setting, in which the holes are limited to appear only in the invariants. Since we allow holes in both the program itself and the invariants, our setting is strictly more general. Moreover, while~\cite{pldi2020} also uses Positivstellens\"atze, it does not consider Nullstellens\"atze, and thus has a different mathematical basis. The work~\cite{sankaranarayanan2004non} considers the same problem as~\cite{pldi2020} and solves it using Gr\"obner basis computations. In comparison with both of these approaches, we handle a more general setting and have a different methodology.  There are also template-based approaches for generation of \emph{linear} loop invariants, such as~\cite{colon2003linear,sankaranarayanan2004constraint,liu2022location}. This is an even narrower special case of our setting. Finally,~\cite{feng2017finding} considers invariant generation for probabilistic programs and uses a different Positivstellensatz, i.e.~that of Stengle, to obtain an automated algorithm. This setting is incomparable to ours due to the presence of probabilistic behavior but absence of holes in the program itself.
	
	\item \emph{Termination analysis:} Template-based algorithms that utilize Farkas' Lemma or Positivstellens\"atze have also been considered in the context of termination and reachability analysis~\cite{neumann2020ranking,agrawal2017lexicographic,chatterjee2016termination,chatterjee2020termination12,asadi2021polynomial,takisaka2018ranking}. However, this is a very different and orthogonal problem and, while there are clear similarities in the approaches, it is not possible to directly or experimentally compare them. Conceptually, the works in termination analysis assume that strong-enough invariants are given as part of the input\footnote{Such invariants can be generated by our approach or~\cite{pldi2020,sankaranarayanan2004non}.} and then reduce the problem of finding an affine/polynomial ranking function, or its probabilistic counterparts, to Linear Programming (LP). This is in contrast to our setting where the reduction is to Quadratic Programming (QP) and a reduction to LP is impossible due to the strong NP-hardness. 
\end{compactitem}

\paragraph{Organization} In Section~\ref{sec:problem-def}, we formalize our problem and define the syntax and semantics of our programs. Section~\ref{sec:algo-new} is the core of the paper in which we provide our synthesis algorithm and prove that it is sound and semi-complete. We then study the complexity of the problem in Section~\ref{sec:complexity}, proving that it is decidable and strongly NP-hard. Finally, Section~\ref{secion_proof_of_concept} reports on an implementation of our algorithm and provides experimental results.

\newcommand{\valuation}{val}
\newcommand{\val}{\valuation}
\newcommand{\vars}{\mathbb{V}}
\newcommand{\tvars}{{\mathbb{T}}}
\newcommand{\tval}{\valuation_\tvars}
\newcommand{\locs}{\mathbb{L}}
\newcommand{\loc}{\ell}
\newcommand{\transitions}{\mathcal{T}}
\newcommand{\transition}{\tau}
\newcommand{\invariant}{\mathbb{I}}
\newcommand{\cutset}{\mathcal{C}}
\renewcommand{\state}{\sigma}

\section{Template-based Synthesis of Polynomial Programs} \label{sec:problem-def}

In this work, we focus on imperative polynomial programs. Our programs can have real variables, conditionals and loops. Moreover, they include specifications as pre and post-conditions, as well as invariants for each loop, which are used to prove that the program satisfies the specification.

\paragraph{Variables and Valuations} We assume that our programs have a finite set $\vars = \{v_1, \ldots, v_k\}$ of real-valued variables. A \emph{valuation} is a function $\valuation: \vars \rightarrow \mathbb R$ that assigns a real value to each program variable.

\paragraph{Syntax} We consider programs that can be generated from the following grammar:
\begin{align*}
	P &:= (\Phi, C,\Phi) \\
	C &:= \textbf{skip} \;\mid\; v \leftarrow \Pi \mid \texttt{ if } (\Phi)\; \{C\} \texttt{ else }\{ C\}  \mid\texttt{while } (\Phi, \Phi) \; \{C\} \mid  C; C\\
	\Phi &:= \Pi \ge 0 \mid (\Phi \land \Phi) \mid \lnot \Phi\\
		v & \in \vars \\
	\Pi & \in \mathbb{R}[\vars]
\end{align*}

Intuitively, a program $P$ consists of a pre-condition, a sequence $C$ of commands, and a post-condition. We have a special command \textbf{skip} that does not do anything. We also allow assignments, conditionals and while loops. The header of each loop should include both a guard and a loop invariant. Moreover, the right-hand-side of every assignment is a polynomial expression $\Pi$ over program variables, and the guards, invariants and pre and post-conditions are boolean combinations $\Phi$ of polynomial inequalities over $\vars.$ For simplicity, we often write the programs as in Figure~\ref{fig:progintrocomplete} to make them more human-readable. We also use standard syntactic sugars, e.g.~to allow other boolean operators.

\paragraph{Polynomial Assertions} A polynomial \emph{assertion} over a set $V$ of variables is a boolean combination of polynomial inequalities of the form $p \geq 0$ in which $p \in \mathbb{R}[V].$ In the sequel, without loss of generality, our algorithms assume that all polynomial assertions are in disjunctive normal form.

\paragraph{Polynomial Transition Systems} We define the semantics of our programs using transition systems. A \textit{polynomial transition system (PTS)} is a tuple $ (\vars,\locs,\loc_0,\theta_0,\loc_f, \theta_f, \transitions, \invariant)$ that consists of:
\begin{compactitem}
	\item A finite set $\vars$ of \emph{program variables},
	\item A finite set $\locs$ of \emph{locations},
	\item An \emph{initial location} $\loc_0 \in \locs$,
	\item An initial polynomial assertion $\theta_0$,
	\item A final location $l_f \in \locs$,
	\item A final polynomial assertion $\theta_f,$
	\item A finite set $\transitions$ of transitions. Each transition $\transition \in \transitions$ is a tuple $(\loc,\loc',\rho_{\transition})$, where $\loc,\loc'\in \locs$ are the pre and post locations, and $\rho_{\transition}$, called the \emph{transition relation}, is a polynomial assertion over $\vars \cup \vars'$, where $\vars$ represents variables in the pre location and its primed version $\vars'$ represents the variables in the post location, i.e.~after taking the transition. Specifically, we have $\vars' = \{v'~\mid~ v \in \vars \};$
	\item An invariant $\invariant$ that maps some of the locations $\loc \in \locs$ to a polynomial assertion $\invariant(\loc).$ Moreover, we have $\invariant(\loc_0) = \theta_0$ and $\invariant(\loc_f) = \theta_f.$ The set of locations for which an invariant is provided must be a cutset (see further below for a formal definition).
\end{compactitem}

The translation from programs to transition systems is intuitive and straightforward. We can have one location $\ell \in \locs$ for each line of the program, and let $\ell_0$ correspond to the initial line and $\ell_f$ to the end of the program. Similarly, $\theta_0$ serves as the pre-condition and $\theta_f$ as the post-condition. Programs are more human-friendly and hence a nicer way to specify the input to the synthesis problem. In contrast, it is easier to define the algorithms using the more formal notion of PTS. As such, we will define our semantics based on a PTS instead of a program. Below, we fix a PTS $(\vars,\locs,\loc_0,\theta_0,\loc_f, \theta_f, \transitions, \invariant).$

\begin{figure}
	\includegraphics[keepaspectratio, width=\linewidth]{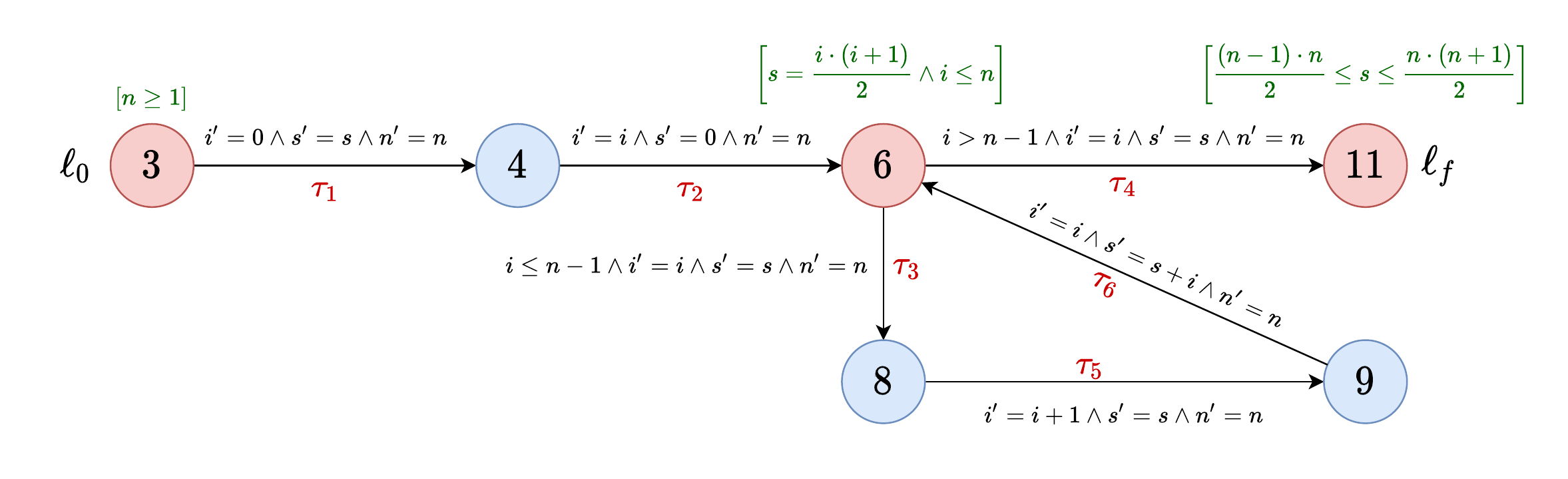}
	\caption{The Polynomial Transition System (PTS) Corresponding to the Program in Figure~\ref{fig:progintrocomplete}.}
	\label{fig:pts}
\end{figure}

\begin{example}
	Figure~\ref{fig:pts} shows the PTS modeling the program of Figure~\ref{fig:progintrocomplete}. We have $\locs = \{3, 4, 6, 8, 9, 11\}$ and each location corresponds to a line of code. The invariants are bracketed and shown in green. Note that our algorithm works with polynomial inequalities, so $a = b$ is syntactic sugar for $a - b \geq 0 \wedge b - a \geq 0.$
\end{example}

\paragraph{Control Flow Graphs (CFGs)} The \emph{control flow graph (CFG)} of our PTS is a directed graph $G = (\locs, E)$ in which the locations serve as vertices and there is a directed edge from $\loc$ to $\loc'$ iff there is a transition $\transition \in \transitions$ of the form $\transition = (\loc, \loc', \rho_\transition).$

\paragraph{Cutsets, Cutpoints and Basic Paths} A subset $\cutset \subseteq \locs$ of vertices of the CFG $G$ is called a \emph{cutset} if
\begin{compactenum}
	\item $\loc_0, \loc_f \in \cutset$; and
	\item  Every cycle in $G,$ be it a simple cycle or not, intersects $\cutset.$
\end{compactenum} 
Assuming we have fixed a cutset $\cutset$, every location $c \in \cutset$ is called a \emph{cutpoint}. A \emph{basic path} is a finite sequence of transitions $\pi = \langle \transition_i \rangle_{i=1}^m = \langle (\loc_i, \loc'_i, \rho_{\transition_i})_{i=1}^m\rangle$ such that:
\begin{compactitem}
	\item $\loc_i = \loc'_{i-1}$ for all $i>1$.
	\item $\loc_1$ and $\loc'_m$ are cutpoints, but no internal location $\loc_i$ with $1 < i \leq m$ is a cutpoint.
\end{compactitem} 
We define the transition relation $\rho_\pi$ of the basic path $\pi$ as $\rho_{\tau_m} \circ \rho_{\tau_{m-1}} \circ \cdots \circ \rho_{\tau_1}.$ Informally, $\rho_\pi$ is the transition relation resulting from going through all the transitions of $\pi$ one-by-one and is hence equal to the sequential composition of $\rho_{\tau_i}$'s. If $\pi$ is the empty path, then we let $\rho_{\pi}$ be the identity relation $\bigwedge_{v \in \vars} v = v'.$

\begin{example} \label{ex:cutset}
	In the PTS of Figure~\ref{fig:pts}, the cutpoints shown in red form a cutset $\cutset = \{3, 6, 11\}.$ In this example, our non-empty basic paths are $\langle \transition_1, \transition_2 \rangle, \langle \transition_4 \rangle,$ and $\langle \transition_3, \transition_5, \transition_6 \rangle.$
\end{example}

\paragraph{Remark on Invariants} It is straightforward to see that it suffices to specify invariants in a cutset only and that the invariants can then be inductively expanded to every other line~\cite{colon2003linear}. In the sequel, we assume that a cutset $\cutset$ is fixed and that the invariant $\invariant$ assigns a polynomial assertion $\invariant(\loc)$ to every cutpoint $\loc \in \cutset.$ This also fits our simple syntax above in which the invariant is only provided for the while loops and the start and end of the program (in the form of pre and post conditions). Nevertheless, our algorithmic approaches consider the PTS and can hence work with any arbitrary cutset $\cutset.$

\paragraph{States} A \emph{state} of the PTS is a pair $\state = (\loc, \valuation)$ in which $\loc \in \locs$ is a location and $\valuation: \vars \rightarrow \mathbb R$ is a valuation. Intuitively, a state specifies which line of the program we are in and what values are taken by our program variables.

\paragraph{Runs} A \emph{run} of our PTS is a sequence $R = \langle \state_i \rangle_{i=0}^m = \langle (\loc_i, \valuation_i) \rangle_{i=0}^m$ of states such that:
\begin{compactitem}
	\item $\val_0 \models \theta_0,$ i.e.~the run starts with the initial location $\loc_0$ and an initial valuation that satisfies the precondition $\theta_0;$
	\item For every $i < m,$ there exists a transition $\transition_i \in \transitions$ that allows us to go from $\sigma_i$ to $\sigma_{i+1}.$ More formally:
	\begin{compactitem}
		\item $\transition_i = (\loc_i, \loc_{i+1}, \rho_{\transition_i});$ and
		\item $(\val_i, \val_{i+1}') \models \rho_{\transition_i},$ i.e.~if we consider $\val_i$ as a valuation on $\vars$ and $\val_{i+1}$ as a valuation on $\vars',$ then they jointly satisfy the transition condition $\rho_{\transition_i}.$
	\end{compactitem}
\end{compactitem} 

\paragraph{Valid Runs} A run $R = \langle \state_i \rangle_{i=0}^m = \langle (\loc_i, \valuation_i) \rangle_{i=0}^m$ is called \emph{valid} iff for every $i$ for which $\loc_i$ has an associated invariant $\invariant(\loc_i),$ we have $\valuation_i \models \invariant(\loc_i).$ Intuitively, a run is valid if it always satisfies the invariant. Specifically, if $\loc_m = \loc_f,$ then we must have $\val_m \models \theta_f$ in order for the run to be valid. 

\paragraph{Valid Transition Systems} A transition system (PTS) is called \emph{valid} iff every run of the PTS is valid.

\paragraph{Inductively Valid Transition System} A PTS is called \emph{inductively valid} with respect to a cutset $\cutset$, iff it satisfies the following conditions for all cutpoints $\loc, \loc' \in \cutset$ and valuations $(\val, \val')$ over $\vars \cup \vars'$:
\begin{compactitem}
	\item \emph{Initiation:} For every basic path $\pi$ from $\loc_0$ to $\loc,$ we have $$\textstyle \val \models \theta_0 \wedge (\val, \val') \models \rho_\pi \Rightarrow \val' \models \invariant(\loc)'.$$ Here, $\invariant(\loc)'$ is the same as $\invariant(\loc),$ except that every program variable $v \in \vars$ is replaced with its primed version $v'.$
	\item \emph{Consecution:} For every basic path $\pi$ from $\loc$ to $\loc',$ we have $$\textstyle \val \models \invariant(\loc) \wedge (\val, \val') \models \rho_\pi \Rightarrow \val' \models \invariant(\loc')'.$$
	\item \emph{Finalization:} For every basic path $\pi$ from $\loc$ to $\loc_f,$ we have $$\textstyle \val \models \invariant(\loc) \wedge (\val, \val') \models \rho_\pi \Rightarrow \val' \models \theta'_f.$$
\end{compactitem}
The intuition behind inductive validity is similar to that of inductive invariants~\cite{colon2003linear}. Informally, a PTS is inductively valid if we can prove that it is valid by breaking each run into a sequence of basic paths and then show its validity by an inductive argument. It is straightforward to see that every inductively valid PTS is indeed valid~\cite{colon2003linear,pldi2020}. Moreover, the finalization and initiation conditions above can be considered as special cases of consecution. We are including them separately to emphasize the point that validity requires the PTS to satisfy its specification (as specified by pre and post-conditions).

\paragraph{Non-determinism} Note that, in order to simplify the presentation, we did not include non-determinism in the syntax of our programs. However, our transition systems have built-in support for non-determinism since one can have several outgoing transitions from the same node or, alternatively, a transition that does not uniquely specify the post-valuation $\valuation'$ in terms of the pre-valuation $\valuation.$ Therefore, all the results in the sequel are trivially extensible to the case of non-deterministic polynomial programs.

Our goal is to synthesize a valid PTS, given a sketch that includes holes in it. To do this, we first define variants of our programs and transition systems which can also include holes.

\paragraph{Sketches of Polynomial Programs}
In this work, we assume that the programmer provides a sketch of the implementation. The sketch is a partial program, i.e.~a program with holes, which is generated using the following grammar: 
\begin{align*}
	P &:= (\Phi, C,\Phi) \\
	C &:= \textbf{skip} \;\mid\; v \leftarrow \Pi \mid v \leftarrow  \textcolor{red}{H} \mid \texttt{ if } (\Phi)\; \{C\} \texttt{ else }\{ C\}  \mid\texttt{while } (\Phi, \Phi) \; \{C\} \mid  C; C\\
	\Phi &:= \Pi \ge 0 \mid \textcolor{red}{H } \ge 0 \mid (\Phi \land \Phi) \mid \lnot \Phi\\
	\textcolor{red}{H} & := (\overline{\vars}, d)\\
	v & \in \vars \\
	\Pi & \in \mathbb{R}[\vars]\\
	\overline{\vars} & \subseteq \vars\\
	d & \in \mathbb N	
\end{align*}
The main difference between this grammar and the previous one is the introduction of a new non-terminal $H$ which models a hole. Note that holes can appear in assignments or as part of the assertions $\Phi,$ be it in guards of if/while or the invariants or even pre and post-conditions. For each hole, the programmer specifies a set $\overline{\vars}$ of variables which she expects might need to appear in filling the hole, as well as a maximum degree $d \in \mathbb N,$ requiring that the hole should be filled by a polynomial of degree at most $d$. As an example, see the sketch in Figure~\ref{fig:progintro} in which the holes are shown in boxes. As usual, we do not follow the grammar exactly in order to make the sketch more human-readable. For example, we wrote the invariant outside the while loop, instead of having it as the second parameter in the while header. 

\paragraph{Template and Symbolic Polynomials and Assertions} Let $\vars$ be our finite set of program variables and $\tvars$ a separate finite set of \emph{template variables}. We define a \emph{symbolic polynomial} over $(\vars, \tvars)$ as a polynomial over $\vars$ whose every coefficient is itself a polynomial over $\tvars.$ More formally, every symbolic polynomial over $(\vars, \tvars)$ is of the form $p \in (\mathbb{R}[\tvars])[\vars].$ We say a symbolic polynomial $p$ is a \emph{template} if every coefficient of $p$, when considered as a polynomial over $\tvars,$ has degree $0$ or $1$. Symbolic (resp.~template) polynomial assertions are simply defined as boolean combinations of symbolic (resp.~template) polynomial inequalities. As mentioned before, our algorithms assume without loss of generality that all assertions are in disjunctive normal form.

\paragraph{Sketches of Polynomial Transition Systems} A \emph{sketch} of a polynomial transition system (SPTS) is tuple $(\vars, \tvars, \locs,\loc_0,\theta_0,\loc_f, \theta_f, \transitions, \invariant)$ in which all parts serve the same purposes and are defined similarly as in a PTS. The only differences are:
\begin{compactitem}
	\item The assertions $\theta_0, \theta_f$ and invariants $\invariant(\loc)$ are now template polynomials over $(\vars, \tvars).$ 
	\item The transition relations $\rho_\transition$ are now template polynomials over $(\vars \cup \vars', \tvars).$
\end{compactitem}
Moreover, we always assume $\tvars \cap \vars = \emptyset.$

\begin{figure}
	\includegraphics[keepaspectratio,width=\linewidth]{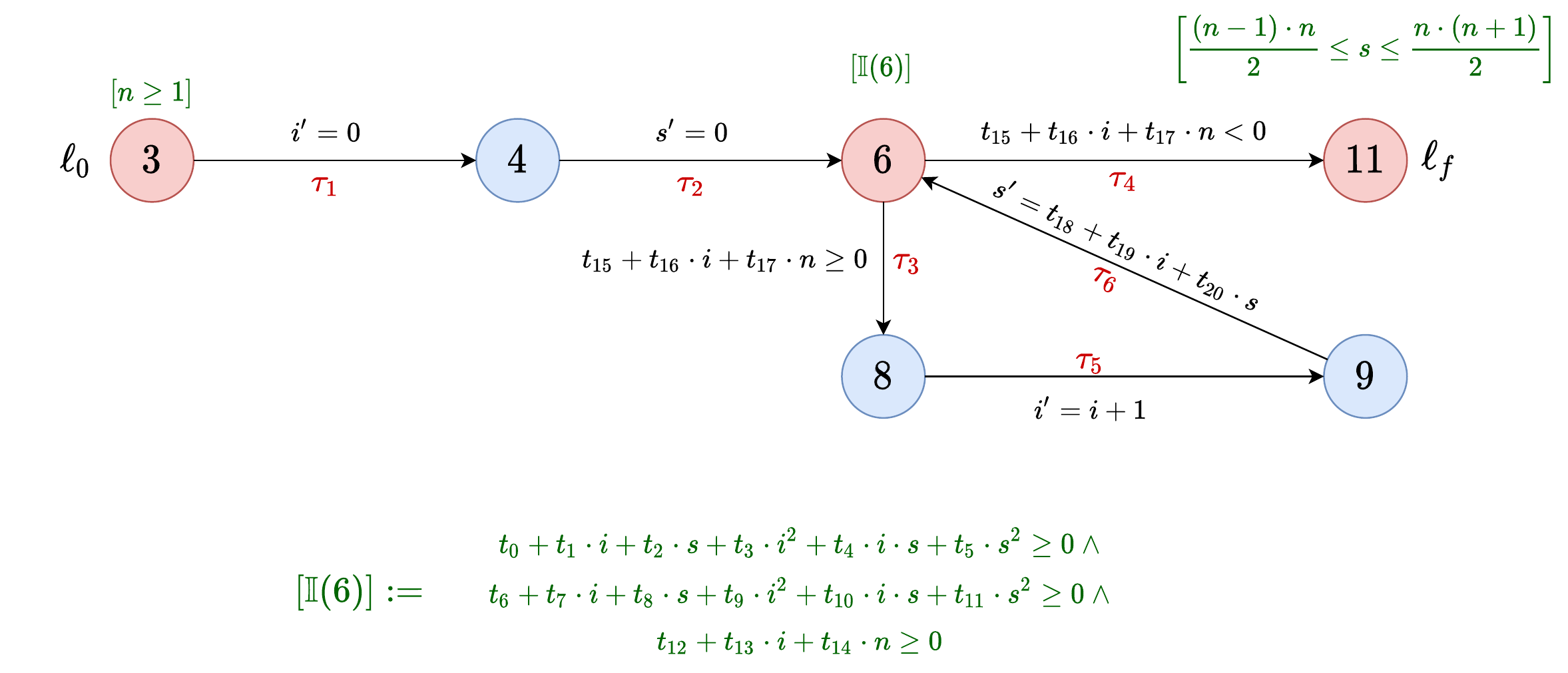}
	\caption{The SPTS Corresponding to the Sketch Program of Figure~\ref{fig:progintro}. For brevity, we have not included equalities of the form $v' = v$ that ensure a variable $v$ does not change in a transition. See Figure~\ref{fig:pts} for details of these equalities. The cutpoints are shown in red.}
	\label{fig:sketchts}
\end{figure}

\begin{example}
	Figure~\ref{fig:sketchts} shows a sketch polynomial transition system (SPTS) corresponding to the sketch program of Figure~\ref{fig:progintro}. To obtain this SPTS, we introduced 21 new template variables $\tvars = \{t_0, \ldots, t_{20}\}$ and replaced each hole with a template polynomial that corresponds to the programmer's desire for that hole. Note that the invariant of location $6$ and the transition relations for $\tau_3, \tau_4$ and $\tau_6$ are now symbolic, i.e.~they are polynomials over the program variables $\{i, s, n\}$ in which every coefficient is itself a polynomial over $\tvars.$ In this case, all coefficients are simply variables in $\tvars.$ Our goal is to find suitable concrete real values for each $t_i$ such that the SPTS of Figure~\ref{fig:sketchts} becomes a valid PTS, e.g.~that of Figure~\ref{fig:pts}.
\end{example}

We now have all the needed ingredients to formalize our \emph{Template-Based Synthesis problem over Polynomial programs (TBSP)}:

\paragraph{TBSP} Given a sketch of a polynomial transition system $\Gamma = (\vars, \tvars, \locs,\loc_0,\theta_0,\loc_f, \theta_f, \transitions, \invariant)$, find a valuation $\tval: \tvars \rightarrow \mathbb R$ that assigns a real value to every template variable, such that $\Gamma^{\tval}$ becomes an inductively valid polynomial transition system. Here, $\Gamma^{\tval}$ is the PTS obtained by taking the sketch $\Gamma$ and replacing every appearance of each template variable $t \in \tvars$ by its corresponding value $\tval(t).$ Informally, our goal is to synthesize concrete values for each of the template variables so that the resulting program/PTS satisfies the programmer's desired specification.

\paragraph{Boundedness} For technical reasons which will become evident in Section~\ref{section:toolkit}, we assume that our program variables $\vars$ have bounded values. More specifically, we consider a bound $M \in (0, \infty)$ and assume that no program variable ever takes a value larger than $M$ or smaller than $-M.$ We call this the \emph{boundedness assumption}. The boundedness assumption is necessary for our completeness theorem (Theorem~\ref{thm:complete}), but not for soundness. It is a merely theoretical assumption with no significant practical effect. In practice, we do not assume boundedness in our implementation in Section~\ref{secion_proof_of_concept} and can successfully handle all the cases without this extra assumption.

\section{Our Algorithm for TBSP} \label{sec:algo-new}

In this section, we provide an automated, sound and semi-complete algorithm for the TBSP problem as defined in Section~\ref{sec:problem-def}. We assume that the input to our algorithm is either a sketch polynomial transition system (SPTS) or a sketch of a polynomial program, containing holes. In the latter case, we would first process the input and translate it to an SPTS. We start by providing a high-level overview of our algorithm (Section~\ref{sec:algo:overview}) and an illustration over the example of Figure~\ref{fig:progintro}. Then, we present our mathematical toolkit, i.e.~Farkas' Lemma, Handelman's Theorem, Putinar's Positivstellesatz and the Real Nullstellensatz in Section~\ref{section:toolkit} and manipulate them to be usable for our application. We use this toolkit to fill in the details of our algorithm in Section~\ref{sec:algo:detail}. Finally, we provide soundness and semi-completeness theorems in Section~\ref{sec:algo:proof}. 

\subsection{An Illustration and Overview of the Algorithm} \label{sec:algo:overview}

Suppose that a sketch, i.e.~a polynomial program with holes, is given as the input. As a running example, let us take Figure~\ref{fig:progintro}. Our algorithm fills in the holes using the following four steps:

\paragraph{Step 1: Creating Templates and Translation to SPTS} For each hole in the sketch, the algorithm generates a template polynomial in which it introduces new template variables. Concretely, let the hole be of the form $H = (\overline{\vars}, d).$ This corresponds to a case in which the programmer desires the hole to be filled by a polynomial of degree at most $d$ over the set $\overline{\vars} \subseteq \vars.$ The algorithm first generates all monomials $m_i$ of degree at most $d$ over $\overline{\vars}.$ It then creates a template polynomial that includes all of these monomials and is of the form $p_H := \sum_{i} t_i \cdot m_i,$ in which each $t_i$ is a newly generated template variable. It replaces the hole $H$ by the template polynomial $p_H$ and then translates the program to an SPTS. The translation from programs to transition systems is a standard procedure, i.e.~a simple parsing exercise, and hence we skip it for brevity. At the end of this step, we have an SPTS $\Gamma = (\vars, \tvars, \locs,\loc_0,\theta_0,\loc_f, \theta_f, \transitions, \invariant).$

\begin{example}
	Consider the sketch of Figure~\ref{fig:progintro}. The first hole in the invariant of line 5 should be filled by a polynomial of degree at most $2$ over the variables $\{i, s\}.$ Hence, the algorithm generates the following template for this hole:
	$$
	\textstyle p_1 := t_0 + t_1 \cdot i + t_2 \cdot s + t_3 \cdot i^2 + t_4 \cdot i \cdot s + t_5 \cdot s^2.
	$$
	Note that all the $t_i$'s are new template variables whose concrete values have to be synthesized later. It then establishes a similar template for the second hole in line 5:
	$$
	\textstyle p_2 := t_6 + t_7 \cdot i + t_8 \cdot s + t_9 \cdot i^2 + t_{10} \cdot i \cdot s + t_{11} \cdot s^2.
	$$
	The third hole in line 5 and the hole in line 6 should be affine expressions over $\{i, n\},$ so the algorithm generates:
	$$
	\begin{matrix}
	p_3 := t_{12} + t_{13} \cdot i + t_{14} \cdot n,\\
	p_4 := t_{15} + t_{16} \cdot i + t_{17} \cdot n.
	\end{matrix}
	$$
	Finally, the last hole (line 9) should be filled by an affine expression over $\{i, s\}.$ So, we have:
	$$
	\textstyle p_5 := t_{18} + t_{19} \cdot i + t_{20} \cdot s.
	$$
	The algorithm first replaces each $i$-th hole with its corresponding template polynomial $p_i.$ It then translates the program to a transition system, leading to the SPTS of Figure~\ref{fig:sketchts}.
\end{example}

\paragraph{Step 2: Generating Entailment Constraints} Recall that our goal is to synthesize concrete real values for the template variables $t_i \in \tvars$ such that our SPTS becomes an inductively valid PTS. As such, it should satisfy the initiation, consecution and finalization constraints defined in Section~\ref{sec:problem-def}. In this step, the algorithm considers every possible basic path $\pi$ and symbolically computes its transition relation and its corresponding initiation/consecution/finalization constraint. Then, the algorithm translates the constraint into an equivalent system of constraints in the following standard form:
\begin{equation} \label{eq:entailment-form}
	f_1 \bowtie_1 0 ~\wedge~ f_2 \bowtie_2 0 ~\wedge~ \ldots \wedge f_r \bowtie_r 0 \Rightarrow f \bowtie 0
\end{equation}
in which $f$ and each $f_i$ are symbolic polynomials over $(\vars \cup \vars', \tvars)$ and each operator $\bowtie, \bowtie_i$ is either $\ge$ or $>$.

\begin{example} \label{ex:step2}
	Consider the SPTS and cutset of Figure~\ref{fig:sketchts}. As discussed in Example~\ref{ex:cutset}, the basic paths in this SPTS are 	
	$\pi_1 = \langle \transition_1, \transition_2 \rangle, \pi_2 = \langle \transition_4 \rangle,$ and $\pi_3 = \langle \transition_3, \transition_5, \transition_6 \rangle.$
	The algorithm symbolically computes the transition relation of each basic path:
	$$
	\rho_{\pi_1} = \rho_{\tau_2} \circ \rho_{\tau_1} = \left[ n'=n ~\wedge~ i'=0 ~\wedge~ s'=0 \right]
	$$
	$$
	\rho_{\pi_2} = \rho_{\tau_4} = \left[ t_{15} + t_{16} \cdot i + t_{17} \cdot n < 0 ~\wedge~ i'=i ~\wedge~ s'=s ~\wedge~ n'=n \right]
	$$
	$$
	\rho_{\pi_3} = \rho_{\tau_6} \circ \rho_{\tau_5} \circ \rho_{\tau_3} = \left[ t_{15} + t_{16} \cdot i + t_{17} \cdot n \geq 0 ~\wedge~ i' = i+1 ~\wedge~ s'=t_{18} + t_{19} \cdot (i+1) + t_{20} \cdot s ~\wedge~ n' = n \right]
	$$
	It then has to symbolically compute the inductive validity constraints. For $\pi_1$, we have the following initiation constraint:
	$$\val \models \theta_0 \wedge (\val, \val') \models \rho_{\pi_1} \Rightarrow \val' \models \invariant(6)'.$$
	The algorithm symbolically computes and expands this constraint and obtains:
	$$
	n \geq 1 ~\wedge~ n' = n ~\wedge~ i' = 0 ~\wedge~ s' = 0 \Rightarrow 
	$$
	$$
	t_0 + t_1 \cdot i' + t_2 \cdot s' + t_3 \cdot i'^2 + t_4 \cdot i' \cdot s' + t_5 \cdot s'^2 \ge 0 ~\wedge~
	$$
	$$
	t_6 + t_7 \cdot i' + t_8 \cdot s' + t_9 \cdot i'^2 + t_{10} \cdot i' \cdot s' + t_{11} \cdot s'^2 \geq 0 ~\wedge~
	$$
	$$
	t_{12} + t_{13} \cdot i' + t_{14} \cdot n' \geq 0
	$$
	Note that this constraint should hold for any valuation of $\vars \cup \vars'.$ At this point, since we know $n' = n$ and $i'=s'=0,$ we can further simplify the constraint to
	$$
	n \geq 1 \Rightarrow t_0 \geq 0 ~\wedge~ t_6 \geq 0 ~\wedge~ t_{12} + t_{14} \cdot \textcolor{red}{n} \geq 0, 
	$$
	hence eliminating several variables and obtaining a much simpler but equivalent constraint. Our implementation in Section~\ref{secion_proof_of_concept} applies such sound heuristics for simplifying the constraints. However, the constraint is still not in the standard form of \eqref{eq:entailment-form} since there are several inequalities on the right-hand-side. The algorithm rewrites this constraint as the following equivalent set of constraints of the standard form \eqref{eq:entailment-form}:
	$$
	n \geq 1 \Rightarrow t_0 \geq 0;
	$$
	$$
	n \geq 1 \Rightarrow t_6 \geq 0;
	$$
	$$
	n \geq 1 \Rightarrow t_{12} + t_{14} \cdot n \geq 0.
	$$
	Let us now consider $\pi_2.$ We need to generate a finalization constraint in this case. The algorithm symbolically computes
	$$
	\invariant(6) ~\wedge~ \rho_{\pi_2} \Rightarrow \theta'_f
	$$
	and then expands it and rewrites it as an equivalent set of standard constraints just as in the previous case. Finally, for $\pi_3,$ which starts and ends at the cutpoint $6$, the algorithm writes a consecution constraint:
	$$
	\invariant(6) ~\wedge~ \rho_{\pi_3} \Rightarrow \invariant(6)'.
	$$
	As above, the constraint will be simplified and translated to an equivalent set of standard constraints.
\end{example}

\paragraph{Step 3: Eliminating Program Variables and Reduction to QP} We start by providing some intuition for this step. Let $\vars = \{v_1, \ldots, v_m\}$ be our program variables and $\tvars = \{t_1, \ldots, t_{m'}\}$ the template variables introduced in Step 1 above. By the end of the previous step, we have a set of standard constraints of form~\eqref{eq:entailment-form}. Let the constraints be $\kappa_1, \kappa_2, \ldots, \kappa_u.$ Our problem is now equivalent to finding concrete real values $t^*_i \in \mathbb R$ for each template variable $t_i \in \tvars$ such that the following formula is satisfied:
\begin{equation} \label{eq:firstorder}
	\forall v_1, v_2, \ldots, v_m ~\;~\;~ (\kappa_1 ~\wedge~ \kappa_2 ~\wedge~ \ldots ~\wedge~ \kappa_u)
\end{equation} 
or more formally
\begin{equation} \label{eq:firstorder22}
	\textstyle \forall \val \in \mathbb{R}^{\vars} ~\;~\;~ \val \models \left( \bigwedge_{j=1}^u \kappa_j \left[ t_i \leftarrow t^*_i \right]_{i=1}^{m'} \right),
\end{equation}
where $\kappa_j \left[ t_i \leftarrow t^*_i \right]_{i=1}^{m'}$ is the formula $\kappa_j$ in which each $t_i$ is replaced by $t^*_i.$ This is not an easy problem in general (See Section~\ref{sec:complexity}), but the fact that all $\kappa_i$ constraints are in the standard form of~\eqref{eq:entailment-form} can be exploited for a reduction to quadratic programming (QP). This step requires mathematical details that will be presented in Section~\ref{section:toolkit}. At this point, we consider a simple sound algorithm that captures the main intuition of this step and then fill in the details to make it semi-complete in Section~\ref{sec:algo:detail}. 

Consider a standard constraint $\kappa_i$ of the following form:
$$
	f_1 \geq 0 ~\wedge~ f_2 \geq 0 ~\wedge~ \ldots \wedge f_r \geq 0 \Rightarrow f \geq  0
$$
As mentioned above, this constraint should hold over all valuations. So, intuitively, we would like to guarantee that whenever all polynomials $f_1, f_2, \ldots, f_r$ are non-negative, then so is $f$. A simple sound solution is to try to write $f$ as a combination of the $f_i$'s in the following form:
\begin{equation} \label{eq:lhsrhs}
\textstyle f = \lambda_0 + \sum_{i=1}^r  \lambda_i \cdot f_i
\end{equation}
where each $\lambda_i$ is a new non-negative real variable. We write $\Lambda$ as a shorthand for $\{\lambda_0, \ldots, \lambda_r\}.$ This is clearly sound, since whenever the $f_i$'s are non-negative at a given valuation, then so is every combination of them with non-negative coefficients. Indeed, we can go one step further and let each $\lambda_i$ be an always-non-negative polynomial, instead of a real constant, and the exact same argument will go through. As we will see, this idea can lead to a semi-complete algorithm. 

We can now present Step 3 of the algorithm. In this step, the algorithm considers each standard constraint $\kappa_i$ and symbolically computes the equation~\eqref{eq:lhsrhs} for it. Note that both sides of~\eqref{eq:lhsrhs} are symbolic polynomials over $(\vars, \tvars \cup \Lambda).$ In particular, they are polynomials over $\vars$. However, two polynomials are equal if and only if they have equal coefficients for every monomial. Hence, the algorithm equates the corresponding coefficients on the left and right-hand sides of~\eqref{eq:lhsrhs} and rewrites it as a system of quadratic equalities over $\tvars \cup \Lambda.$

\begin{example}
	Consider the standard constraint 
	$
	n - 1 \geq 0 \Rightarrow t_{12} + t_{14} \cdot n \geq 0
	$ that was generated in Step 2 above (Example~\ref{ex:step2}). The algorithm symbolically computes~\eqref{eq:lhsrhs} and obtains
	$$
	t_{12} + t_{14} \cdot n = \lambda_0 + \lambda_1 \cdot (n-1).
	$$
	Since both sides of this equality are polynomials over $n$, in order for the equality to hold, they should have the same constant factor and the same coefficient for $n$, so the algorithm obtains the following constraints:
	$$
	t_{12} = \lambda_0 - \lambda_1
	$$
	$$
	t_{14} = \lambda_1
	$$
	$$
	\lambda_0, \lambda_1 \geq 0
	$$
	This is a quadratic programming (QP) instance in which none of the program variables, e.g.~$n$, appear. The algorithm translates other standard constraints to QP in a similar manner and combines them conjunctively.
\end{example}

\paragraph{Step 4: Constraint Solving} The algorithm passes the QP instance generated in the previous step to an SMT solver or a numerical solver. If the solver fails to find a solution, then the algorithm fails, too. Otherwise, the solver provides values for all $t_i$ and $\lambda_i$ variables. The algorithm plugs the values found for $t_i$'s back into the templates of Step 1, which successfully leads to an inductively valid PTS or equivalenty to a valid synthesized program that satisfies the desired specification.

\begin{example}
	Given the sketch of Figure~\ref{fig:progintro} as input, the synthesized program of Figure~\ref{fig:progintrocomplete} is obtained as one of the solutions of the QP instance in Step 4. However, the solution is not necessarily unique, and our algorithm might find a different solution based on the values found by the external solver.
\end{example}

\newcommand{\sat}{\textsl{SAT}}
\newcommand{\monoid}{\textsl{Monoid}}

\subsection{Mathematical Toolkit}	\label{section:toolkit}

In Section~\ref{sec:algo:overview} above, we provided an overview of our algorithm. However, the details of Step 3, i.e.~the reduction to QP, were not presented since they depend on certain mathematical prerequisites. In this section, we provide the mathematical tools and theorems that are crucial for this step of the algorithm. We first recall some notation and classical definitions. Then, we present several theorems from polyhedral and real algebraic geometry. Finally, we obtain tailor-made versions of these theorems in a format that can be used in Step 3 of our algorithm above. We refer to \cite{hartshorne2013algebraic,bochnak2013real} for a more detailed treatment of these theorems.

\paragraph{Sums of Squares}
A polynomial $h \in \mathbb{R}[x_1,\dots,x_n]$ is a \emph{sum of squares (SOS)} iff there exist $k \geq 1$ and polynomials $g_1, \dots, g_k \in \mathbb{R}[x_1, \dots, x_n]$ such that $h = \sum_{i=1}^k g_i^2.$

\paragraph{Strong Positivity}
Given a set $X \subseteq \mathbb{R}^n$ and a polynomial $g \in \mathbb{R}[x_1,\dots,x_n]$, we say that $g$ is strongly positive over X if $\inf_{x\in X} g(x) > 0$. We write this as $X \models g \gg 0.$

\paragraph{Notation} Let $\Phi = \{g_1 \bowtie_1 0,\dots, g_k \bowtie_k 0\}$ be a set of polynomial inequalities where $\bowtie_i \in \{\geq,>\}$ and $g_i \in \mathbb{R}[x_1, \dots, x_n]$. We define $\sat(\Phi)$ as the set of all real valuations $\val$ over the variables $\{x_1,\dots,x_n\}$ that satisfy $\Phi.$ More formally, $\sat(\Phi) = \{ \val \in \mathbb{R}^n ~\vert~ \bigwedge_{i=1}^k  \val \models (g_i \bowtie_i 0)\}.$

\paragraph{Closure}
Given a set $X \subseteq \mathbb{R}^n$, we define $\overline{X}$ to be the closure of $X$ with respect to the Euclidean topology of $\mathbb{R}^n$. For a set $\Phi$ of polynomial inequalities, we define $\overline{\Phi}$ as the system of polynomial inequalities obtained from $\Phi$ by replacing every strict ineqaulity with its non-strict counterpart.

We are now ready to present the main mathematical theorems that will be used in our work. Our presentation follows that of~\cite[Section 2.6]{goharshady2020parameterized}, which also contains proofs of corollaries that are not proven here.

\begin{theorem}[{Farkas' Lemma \cite{farkas1902theorie}}]\label{thm:Farkas'}
Consider a set $V = \{x_1, \ldots, x_r\}$ of real-valued variables and the following system $\Phi$ of equations over $V$:
    \begin{equation*}
        \Phi := 
        \begin{cases}
            a_{1,0} + a_{1,1} \cdot x_1 + \ldots + a_{1,r} \cdot x_r \ge 0\\
            \hfil \vdots\\
            a_{m,0} + a_{m,1} \cdot x_1 + \ldots + a_{m,r} \cdot x_r \ge 0\\
        \end{cases}.
    \end{equation*}
     When $\Phi$ is satisfiable, it entails a linear inequality $$\psi := c_0 + c_1 \cdot x_1 + \dots + c_r \cdot x_r \ge 0$$
    if and only if $\psi$ can be written as non-negative linear combination of the inequalities in $\Phi$ and the trivial inequality $1 \ge 0,$ i.e.~if there exist non-negative real numbers $y_0,\dots,y_m$ such that 
    $$
    \begin{matrix}
    c_0 = y_0 + \sum_{i=1}^k y_i \cdot a_{i, 0};\\
    c_1 = \sum_{i=1}^k y_i \cdot a_{i, 1};\\
    \vdots\\
    c_r = \sum_{i=1}^k y_i \cdot a_{i, r}.
    \end{matrix}
    $$
    Moreover, $\Phi$ is unsatisfiable if and only if $-1 \ge 0$ can be derived as above. \
\end{theorem}

The importance of Farkas' Lemma for us is that if we have a standard constraint of form~\eqref{eq:entailment-form} and if the constraint includes only linear/affine inequalities, then we can use this lemma in Step 3 of our algorithm to reduce the standard constraint to QP, just as we did in Section~\ref{sec:algo:overview}. Moreover, Farkas' Lemma guarantees that this approach is not only sound but also complete. A corner case that we have to consider is when $\Phi$ is itself unsatisfiable and thus $\Phi \Rightarrow \psi$ holds vacuously. Fortunately, Farkas' Lemma also provides a criterion for unsatisfiability. In practice, we work with the following corollary of Theorem~\ref{thm:Farkas'} which can also handle strict inequalities.
\begin{corollary}\label{thm:farkas_2}
Consider a set $V = \{x_1, \ldots, x_r\}$ of real-valued variables and the following system of equations over $V$:
        \begin{equation*}
        \Phi := 
        \begin{cases}
            a_{1,0} + a_{1,1} \cdot x_1 + \ldots + a_{1,r} \cdot x_r \bowtie_1  0\\
            \hfil \vdots\\
            a_{m,0} + a_{m,1} \cdot x_1 + \ldots + a_{m,r} \cdot x_r \bowtie_m 0\\
        \end{cases}
    \end{equation*}
    where $\bowtie_i \in \{>,\geq\}$ for all $1 \leq i \leq m$. When $\Phi$ is satisfiable, it entails a linear inequality $$\psi := c_0 + c_1 \cdot x_1 + \dots + c_r \cdot x_r \bowtie 0$$
    with $\bowtie \in \{>, \geq\}$, if and only if $\psi$ can be written as non-negative linear combination of inequalities in $\Phi$ and the trivial inequality $1 > 0$. Note that if $\psi$ is strict, then at least one of the strict inequalities should appear with a non-zero coefficient in the linear combination. Moreover, $\Phi$ is unsatisfiable if and only if either $-1 \ge 0$ or $0 > 0$ can be derived as above.
\end{corollary}

We now consider extensions of Farkas' Lemma which can help us handle non-linear standard constraints in Step 3. The first extension is Handelman's theorem, which can be applied when the inequalities on the left hand side of~\eqref{eq:entailment-form} are linear/affine, but the right hand side is a polynomial of arbitrary degree. To present the theorem, we first need the concept of monoids.

\paragraph{Monoid}
Consider a set $V = \{x_1,\dots x_r \}$ of real-valued variables and the following system of linear inequalities over $V$:
        \begin{equation*}
        \Phi := 
        \begin{cases}
            a_{1,0} + a_{1,1} \cdot x_1 + \ldots + a_{1,r} \cdot x_r \bowtie_1  0\\
            \hfil \vdots\\
            a_{m,0} + a_{m,1} \cdot x_1 + \ldots + a_{m,r} \cdot x_r \bowtie_m 0\\
        \end{cases}
    \end{equation*}
    where $\bowtie_i \in \{>, \geq \}$ for all $1 \leq i \leq m$. Let $g_i$ be the left hand side of the $i$-th inequality, i.e. $g_i(x_1,\dots,x_r) := a_{i,0} + a_{i, 1} \cdot x_1 + \dots a_{i,r} \cdot x_r$. The \emph{monoid} of $\Phi$ is defined as: 
   	$$
    \textstyle \monoid(\Phi) := \left\{ \prod_{i=1}^m g_i^{k_i} \mid m \in \mathbb{N} ~\wedge~\forall i ~\; k_i \in \mathbb{N} \cup \{0\} \right\}.
    $$
In other words, the monoid contains all polynomials that can be obtained as a multiplication of the $g_i$'s. Note that $1 \in \monoid(\Phi).$ 
We define $\monoid_d(\Phi)$ as the subset of polynomials in $\monoid(\Phi)$ of degree at most $d.$

\begin{theorem}[{Handelman's Theorem \cite{handelman1988representing}}]\label{thm:handel}
Consider a set $V = \{x_1, \dots, x_r \}$ of real-valued variables and the following system of equations over $V$:
     \begin{equation*}
        \Phi := 
        \begin{cases}
            a_{1,0} + a_{1,1} \cdot x_1 + \ldots + a_{1,r} \cdot x_r \geq  0\\
            \hfil \vdots\\
            a_{m,0} + a_{m,1} \cdot x_1 + \ldots + a_{m,r} \cdot x_r \geq 0\\
        \end{cases}.
    \end{equation*}
If $\Phi$ is satisfiable, $\sat(\Phi)$ is compact, and $\Phi$ entails a polynomial inequality $g(x_1,\dots,x_r) > 0,$ then there exist non-negative real numbers $y_1,\dots y_u$ and polynomials $h_1,\dots,h_u \in \monoid(\Phi)$ such that:
$$
\textstyle g = \sum_{i=1}^u y_i \cdot h_i.
$$
\end{theorem}

The intuition here is that if every inequality in $\Phi$ holds, then all the LHS expressions in $\Phi$ are non-negative and hence any multiplication $h_i$ of them is also non-negative. As in the case of Farkas' Lemma, Handelman's theorem shows that this approach is not only sound but also complete. We also need a variant that can handle strict inequalities in $\Phi.$

\begin{corollary}\label{thm:handel_2}
Consider a set $V = \{x_1, \dots, x_r \}$ of real-valued variables and the following system of equations over $V$:
     \begin{equation*}
        \Phi := 
        \begin{cases}
            a_{1,0} + a_{1,1} \cdot x_1 + \ldots + a_{1,r} \cdot x_r \bowtie_1  0\\
            \hfil \vdots\\
            a_{m,0} + a_{m,1} \cdot x_1 + \ldots + a_{m,r} \cdot x_r \bowtie_m 0\\
        \end{cases}
    \end{equation*}
in which $\bowtie_i \in \{>,\geq\}$ for all $1 \leq i \leq m$. If $\Phi$ is satisfiable and $\sat(\Phi)$ is bounded, then $\Phi$ entails a strong polynomial inequality $g \gg 0 $ if and only if there exist constants $y_0 \in (0,\infty),$ and $y_1,\dots,y_u \in [0,\infty)$, and polynomials $h_1,\dots,h_u \in \monoid(\Phi)$ such that:
$$
\textstyle g = y_0 + \sum_{i=1}^u y_i \cdot h_i.
$$
\end{corollary}

Corollary~\ref{thm:handel_2} above can handle a wider family of standard constraints than Corollary~\ref{thm:farkas_2}. However, it is also more expensive, since we now need to generate one new variable $y_i$ for every polynomial in $\monoid_d(\Phi)$ instead of $\Phi$ itself. Moreover, there is no bound $d$ in the theorem itself, so introducing $d$ would lead to semi-completeness instead of completeness, i.e.~the approach would be complete only if a large enough value of $d$ is used. As such, in cases where both sides of~\eqref{eq:entailment-form} are linear, Corollary~\ref{thm:farkas_2} is preferable. We now consider a more expressive theorem that can handle polynomials on boths sides of~\eqref{eq:entailment-form}.

\begin{theorem}[Putinar's Positivstellensatz \cite{putinar1993positive}]\label{theorem_put_pos}
Given a finite collection of polynomials $\{g,g_1,\dots,g_k \} \in \mathbb{R}[x_1,\dots,x_n]$, let $\Phi$ be the set of inequalities defined as
$$
\Phi: \{g_1 \geq 0,  \dots , g_k \geq 0 \}.
$$
If $\Phi$ entails the polynomial inequality $g > 0$ and there exist some $i$ such that $\sat(g_i \geq 0)$ is compact, then there exist polynomials $h_0, \dots, h_k \in \mathbb{R}[x_1,\dots,x_n]$ such that 
$$
\textstyle g = h_0 + \sum_{i=1}^{m} h_i \cdot g_i 
$$
and every $h_i$ is a sum of squares. Moreover, $\Phi$ is unsatisfiable if and only if $-1 > 0$ can be obtained as above, i.e.~with $g = -1$.
\end{theorem}

As in the cases of Farkas and Handelman, we need a variant of of Theorem~\ref{theorem_put_pos} that can handle strict inequalities in $\Phi.$

\begin{corollary}\label{thm_pos_put_2}
Consider a finite collection of polynomials $\{g,g_1,\dots,g_k \} \in \mathbb{R}[x_1,\dots,x_n]$ and let
$$
\Phi: \{g_1 \bowtie_1 0,  \dots , g_k \bowtie_k 0 \}
$$
where $\bowtie_i \in \{>,\geq \}$ for all $1 \leq i \leq k$. Assume that there exist some $i$ such that $\sat(g_i \geq 0)$ is compact or equivalently $\sat(g_i \bowtie_i 0)$ is bounded. If $\Phi$ is satisfiable, then it entails the strong polynomial inequality $g \gg 0$, iff there exists a constant $y_0 \in (0, \infty)$ and  polynomials $h_0, \dots, h_k \in \mathbb{R}[x_1,\dots,x_n]$ such that 
$$
\textstyle g = y_0 + h_0 + \sum_{i=1}^{k} h_i \cdot g_i, 
$$
and every $h_i$ is a sum of squares.
\end{corollary}

Trying to use the corollary above for handling standard constraints of form~\eqref{eq:entailment-form} in Step 3 of our algorithm leads to two problems: (i)~we also need a criterion for unsatisfiability of $\Phi$ to handle the cases where $\Phi \Rightarrow \psi$ holds vacuously, and (ii)~in our QP, we should somehow express the property that every $h_i$ is a sum of squares. We now show how each of these challenges can be handled. To handle (i), we need another classical theorem from real algebraic geometry.

\begin{theorem}[{\cite[Corollary 4.1.8]{bochnak2013real}\label{theorem_real} the Real Nullstellensatz}]
Given polynomials $g,g_1,\dots,g_k \in \mathbb{R}[x_1,\dots,x_n]$, exactly one of the following two statements holds:
\begin{compactitem}
    \item There exists $x \in \mathbb{R}^n,$ such that $g_1(x) = \dots= g_k(x) = 0$, but $g(x) \neq 0$.
    \item There exists $\alpha \in \mathbb{N}\cup \{0\}$ and polynomials $h_1,\dots,h_k \in \mathbb{R}[x_1,\dots,x_n]$ such that $\sum_{i=1} ^k h_i \cdot g_i - h_0= g^{2 \cdot \alpha}$ and $h_0$ is a sum of squares.
\end{compactitem}
\end{theorem}

We now combine the Real Nullstellensatz with Putinar's Postivstellensatz to obtain a criterion for unsatisfiability of $\Phi.$

\begin{theorem}\label{technical_tool_1}
Consider a finite collection of polynomials $\{g_1,\dots,g_k \} \in \mathbb{R}[x_1,\dots,x_n]$ and the following system of inequalities:
\[
\Phi: \{g_1 \bowtie_1 0,  \dots , g_k \bowtie_k 0 \}
\]
where $\bowtie_i \in \{>,\geq \}$ for all $1 \leq i \leq k$. $\Phi$ is unsatisfiable if and only if at least one of the following statements holds:
\begin{compactitem}
    \item There exist a constant $y_0 \in (0,\infty)$ and sum of square polynomials $h_0,\dots,h_k \in \mathbb{R}[x_1,\dots,x_n]$ such that
    \begin{equation}\label{eqn:1}
     \textstyle     -1 = y_0 + h_0 + \sum_{i=1}^k h_i \cdot g_i.
    \end{equation}
    \item There exist $\alpha \in \mathbb{N} \cup \{0\}$ and $h_0, h_1,\dots,h_k \in \mathbb{R}[x_1,\dots,x_n,w_1,\dots,w_k]$, such that for some $1 \leq j \leq m $ with $\bowtie_j \in \{> \}$, we have 
    \begin{equation}\label{eqn:2}
\textstyle    w_j ^{4 \cdot \alpha} = \sum_{i=1} ^m h_i  \cdot (g_i - w_i ^2) - h_0.
    \end{equation}

    where $h_0$ is a sum of squares in $\mathbb{R}[x_1,\dots,x_n]$. Note that $w_1, \dots, w_k$ are new variables.
\end{compactitem}
\end{theorem}
\begin{proof}
First we show that if any of the two equalities \eqref{eqn:1} or \eqref{eqn:2} holds then $\Phi$ is unsatisfiable. Suppose $\Phi$ is satisfiable and pick $\val \in \sat(\Phi)$. Then, the RHS of~\eqref{eqn:1} is positive at $\val,$ whereas the LHS is negative. So,~\eqref{eqn:1} cannot hold. Now define $\tilde{g_i}(x_1,\dots,x_n,w_1,\dots,w_k) :=  g_i(x_1,\dots,x_n) - w_i^2.$  Using this definition, we can rewrite \eqref{eqn:2} as $w_j ^{4.\alpha} = \sum_{i=1} ^m h_i \cdot \tilde{g_i} - h_0$. Moreover, $g_j^{2 \cdot \alpha} = (\tilde{g_j} + w_j^2)^{2 \cdot \alpha}.$ Expanding the right hand side using the binomial theorem, we get $g_j^{2 \cdot \alpha} = w_j^{4 \cdot \alpha} + h_j'\cdot \tilde{g_j}$ for some $h_j' \in \mathbb{R}[x_1,\dots,x_n,w_1,\dots,w_k]$. Now, substituting $w_j^{4 \cdot \alpha}$ \eqref{eqn:2}, we get
\begin{equation}\label{eqn:3}
\textstyle g_j^{2 \cdot \alpha} = \sum_{i=1}^m h_i \cdot \tilde{g_i} - h_0 + h_j' \cdot \tilde{g_j} 
\end{equation}
Let us extend $\val$, which is a valuation of $\{v_1,\dots ,v_n\}$ to a valuation $\val'$ over $\{v_1,\dots,v_n,w_1,\dots,w_k\}$ such that $\val' \models \tilde{g_i}(v_1,\dots,v_n,w_1,\dots,w_k) = 0$ for all $1 \leq i \leq k$. Note that such an extension is always possible. We get a contradiction by evaluating \eqref{eqn:3} on $\val'$ as the LHS is positive, whereas the RHS is negative.

We now prove the other side. Suppose $\Phi$ is unsatisfiable. We have two possibilities: either $\overline{\Phi}$ is also unsatisfiable or $\overline{\Phi}$ is satisfiable. Suppose $\overline{\Phi}$ is unsatisfiable, then using Theorem \ref{theorem_put_pos}, $\overline{\Phi}$ entails $-2 > 0$ and we can write $-2 = h_0 + \sum_{i=1}^k h_i \cdot g_i$ for sum of squares polynomials $h_0,\dots,h_k \in \mathbb{R}[x_1,\dots,x_n]$. Therefore,
$$
\textstyle -1 = 1 + h_0 + \sum_{i=1}^k h_i \cdot g_i
$$
which fits into \eqref{eqn:1}. Now we are left with the case when $\overline{\Phi}$ is satisfiable but $\Phi$ is unsatisfiable. We first reorder the inequalities in $\Phi$ such that the non-strict inequalities appear first in the order. Let $j$ be the smallest index for which $\Phi[1\dots j]$, i.e. the set of first $j$ inequalities in $\Phi$, is unsatisfiable. By definition, $\Phi[1\dots j-1]$ is satisfiable and hence $\overline{\sat(\Phi[1\dots j-1])} = \sat(\overline{\Phi}[1\dots j-1])$. We can rewrite $\Phi[1\dots j] = \Phi[1\dots j-1]\wedge (g_j > 0) $. As $\Phi[1\dots j] $  is unsatisfiable, we know that $\Phi[1\dots j-1] $ entails $g_j \leq 0$. More precisely, this means $\sat(\Phi[1\dots j-1]) \subseteq \sat(g_j \leq 0)$. On taking closures on both sides we get $\emph{SAT}( \overline{\Phi}[1\dots j-1]) \subseteq \emph{SAT}(g_j \leq 0)$. This implies that $\overline{\Phi}[1,\dots j-1]$ entails $g_j \leq  0$ and hence $\overline{\Phi}[1\dots j]$ entails $g_j = 0$. As defined above, $\tilde{g_i}(x_1,\dots,x_n,w_1,\dots,w_k) = g_i(x_1,\dots,x_n) - w_i^2$. Now, we will show that there is no valuation $\val^*$ over the variables $\{x_1,\dots, x_n,w_1,\dots,w_k\}$ such that for all $1\leq i \leq j$, $\tilde{g_i}(\val^*) = 0$ but $g_j(\val^*) \neq 0$. Suppose there exist such a valuation $\val^*$. Let us define $\val$ to be the restriction of $\val^*$ to $\{x_1,\dots,x_n\}$. For each $1 \leq i \leq j $, we get $g_i(\val) \geq  0$ as $\tilde{g_j}(\val^*) = 0$. We also get $g_j(\val) = g_j(\val^*) \neq 0 $. Hence, we get a contradiction with the previous result that $\overline{\Phi}[1\dots j] $ entails $g_j = 0$. Applying the Real Nullstellensatz (Theorem \ref{theorem_real}) to $\tilde{g_i}$'s and $g_j$, we have
$$
\textstyle g_j^{2\cdot \alpha} = \sum_{i=1}^j \tilde{h_i} \cdot \tilde{g_i} - h_0
$$
where $\alpha$ is a non-negative integer and $\tilde{h_i}$'s and $h_0$ are sums of sqaures in $\mathbb{R}[x_1,\dots,x_n,w_1,\dots,w_k]$. Using the definition of $g_j = \tilde{g_j} + w_j^2$ and the binomial theorem, we get $g_j^{2 \cdot \alpha} = w_j^{4 \cdot \alpha} + h_j' \cdot \tilde{g_j}$ for some $h_j' \in \mathbb{R}[x_1,\dots,x_n,w_1,\dots,w_k]$. Therefore, we finally get the following expression:
$$
\textstyle w_j^{4 \cdot \alpha} = \sum_{i=1}^j \tilde{h_i} \cdot (g_i-w_i^2) - h_j' \cdot (g_j - w_j^2) - h_0
$$
which fits into the format of~\eqref{eqn:2}, hence completing the proof.
\end{proof}

Finally, we provide the needed theorems to check that a certain polynomial $h$ is a sum of squares using QP.

\begin{theorem}[{\cite[Theorem 3.39]{blekherman2012semidefinite}}]\label{thm:sos_semi}
	
	Let $\vec{a}$ be the vector of all $\binom{n+d}{d}$ monomials of degree less than or equal to $d$ over the variables $\{x_1, \dots, x_n\}.$ A polynomial $p \in \mathbb{R}[x_1,\dots,x_n]$ of degree $2 \cdot d$ is a sum of squares if and only if there exist a positive semidefinite matrix $Q$ of order $\binom{n+d}{d}$ such that $p = a^T \cdot Q \cdot a.$
\end{theorem}

\begin{theorem}[Cholesky decomposition \cite{watkins2004fundamentals}]\label{thm:chole}
A symmetric square matrix $Q$ is positive semidefinite if and only if it has a Cholesky decomposition of the form $Q = L L^T$ where $L$ is a lower-triangular matrix with non-negative diagonal entries.
\end{theorem}

Based on Theorems~\ref{thm:sos_semi} and~\ref{thm:chole}, a polynomial $p$ of degree $2 \cdot d$ is a sum of squares if and only if it can be written as
$
p = a^T \cdot L \cdot L^T \cdot a
$
such that diagonal entries of $L$ are non-negative. This representation provides us with a simple approach to generate a sum-of-squares polynomial of degree $2 \cdot d$ with symbolic coefficients and encoding them in QP. We first generate a lower triangular matrix $L$ of order $ \binom{n+d}{d}$ by creating one fresh variable for each entry in the lower triangle and adding the extra condition that the entries on the diagonal must be non-negative. Then, we symbolically compute $a^T \cdot L \cdot L^T \cdot a.$

\subsection{Details of Step 3 of the Algorithm}
\label{sec:algo:detail}

We now have all the necessary ingredients to provide a variant of Step 3 of the algorithm that preserves completeness. We assume a positive integer $d$ is given as part of the input. This $d$ serves as an upper-bound on the degrees of polynomials/templates that we use in Handelman's theorem (the monoid) and the Stellens\"atze. Our approach is complete as long as a large enough $d$ is chosen.

\paragraph{Step 3: Eliminating Program Variables and Reduction to QP} Recall that at the end of Step 2, we have a finite set of standard constraints of form~\eqref{eq:entailment-form}. In this step, the algorithm handles each standard constraint separately and reduces it to quadratic programming over template variables and newly-introduced variables, hence effectively eliminating the program variables and the quantification over them. Let 
$
	f_1 \bowtie_1 0 ~\wedge~ f_2 \bowtie_2 0 ~\wedge~ \ldots \wedge f_r \bowtie_r 0 \Rightarrow f \bowtie 0
$ be one of the standard constraints. The algorithm considers three cases: (i)~if all the inequalities on both sides of the constraint are affine, then it applies Farkas' Lemma; (ii) if the LHS inequalities are affine but the RHS is a higher-degree polynomial, then the algorithm applies Handelman's theorem; and (iii) if the LHS contains higher-degree polynomials, the algorithm applies the Stellens\"atze and Theorem~\ref{technical_tool_1}. Below, we define $\Phi : \{ f_1 \bowtie_1 0, f_2 \bowtie_2 0, \ldots , f_r \bowtie_r 0 \}$ and $\psi: f \bowtie 0.$

\paragraph{Step 3.(i). Applying Farkas' Lemma} Assuming all the constraints in $\Phi$ and $\psi$ are affine, we can apply Corollary~\ref{thm:farkas_2}. Based on this corollary, we have to consider three cases disjunctively:
    \begin{compactenum}
        \item \emph{$\Phi$ is satisfiable and entails $\psi$}: The algorithm creates new template variables $y_0,\dots,y_r$ with the constraint $y_i \geq 0$ for every $i$. It then symbolically computes
        $
        f = y_0 + \sum_{i=1}^r y_i \cdot f_i.
        $ The latter is a polynomial equality over $\vars.$ So, the algorithm equates the coefficients of corresponding monomials on both sides, hence reducing this case to QP. Additionally, if $\psi$ is a strict inequality, the algorithm adds the extra constraint $\sum_{\bowtie_i \in \{ > \}} y_i > 0.$
        
        \item \emph{$\Phi$ is unsatisfiable and $-1 \geq 0$ can be obtained}: This is similar to the previous case, except that $-1$ should be written as $y_0 + \sum_{i=1}^r y_i \cdot f_i.$
        
        \item \emph{$\Phi$ is unsatisfiable and $0 > 0$ can be obtained}: This is also similar to the last two cases. We have $0 = y_0 + \sum_{i=1}^r y_i \cdot f_i$ and $\sum_{\bowtie_i \in \{ > \}} y_i > 0.$
    \end{compactenum}
Note that the template variables $y$ are freshly generated in each case above. Also, we have to consider cases (2) and (3) because $\Phi$ is unsatisfiable in these cases and hence the constraint $\Phi \Rightarrow \psi$ always holds vacuously.

\paragraph{Step 3.(ii). Applying Handelman's Theorem} Assuming all constraints in $\Phi$ are linear but $\psi$ is a higher-degree polynomial inequality, the algorithm applies Corollaries~\ref{thm:farkas_2} and~\ref{thm:handel_2}. Again, we have to consider the same three cases as in Step 3.(i):
\begin{compactenum}
\item \emph{$\Phi$ is satisfiable and entails $\psi$}: We apply Corollary~\ref{thm:handel_2}. The algorithm first symbolically computes $\monoid_d(\Phi) = \{h_1, h_2, \dots, h_u\}.$ It then generates new template variables $y_0, y_1, \dots, y_u$ and constrains them by setting $y_0, y_1, y_2, \dots, y_u \geq 0.$ If $\psi$ is a strict inequality, it further adds the constraint $y_0 > 0.$ It then symbolically computes the equality
$$
\textstyle f = y_0 + \sum_{i=1}^u y_i \cdot h_i.
$$
As usual, this is an equality whose both sides are polynomials over $\vars.$ So, the algorithm equates the coefficients of corresponding monomials on the LHS and RHS, which reduces this case to QP.
\item Note that $\Phi$ consists of linear inequalities, so we can use Farkas' Lemma to check if $\Phi$ is unsatisfiable. As such, this step is the same as case (2) of Step 3.(i).
\item This is the same as case (3) of Step 3.(i).
\end{compactenum}

\paragraph{Step 3.(iii). Applying Stellens\"atze} 
If $\Phi$ includes polynomial inequalities of degree $2$ or larger, then we have to apply Corollary~\ref{thm_pos_put_2} and Theorem~\ref{technical_tool_1}. The algorithm considers three cases and combines them disjunctively:
\begin{compactenum}
	 \item \emph{$\Phi$ is satisfiable and entails $\psi$}: In this case, we apply Corollary~\ref{thm_pos_put_2}. The algorithm generates template sum-of-squares polynomials $h_0, \ldots, h_r$ of degree $d$ and adds QP constraints that ensure each $h_i$ is a sum of squares (See the end of Section~\ref{section:toolkit}). It also generates a non-negative fresh variable $y_0$. If $\psi$ is strict, the algorithm adds the constraint $y_0 > 0.$ Finally, the algorithm symbolically computes
	 $$
\textstyle	 f = y_0 + h_0 + \sum_{i=1}^r h_i \cdot f_i;
	 $$
	 and equates the corresponding coefficients in the LHS and RHS to obtain QP constraints.
	 
	 \item \emph{$\Phi$ is unsatisfiable due to the first condition of Theorem~\ref{technical_tool_1}}: This case is handled similary to case (1) above, except that we have $-1 = y_0 + h_0 + \sum_{i=1}^r h_i \cdot f_i.$
	 
	 \item \emph{$\Phi$ is unsatisfiable due to the second condition of Theorem~\ref{technical_tool_1}}: The algorithm introduces $r$ new \emph{program variables} $w_1, \dots, w_r.$ It then generates a sum-of-squares template polynomial $h$ over $\vars$ and arbitrary template polynomials $h_1, \dots, h_k$ over $\vars \cup \{w_1, \dots, w_r\}.$ All $h_i$'s are degree-$d$ templates. Finally, for every index $j$ that corresponds to a strict inequality, i.e.~$\bowtie_j \in \{ > \},$ the algorithm symbolically computes
	 $$
\textstyle	 w_j^{d'} = \sum_{i=1}^r h_i \cdot (f_i - w_i^2) - h_0,
	 $$
	 in which $d'$ is the largest multiple of $4$ that does not exceed $d$. Note that this is an equality between two polynomials over $\vars \cup \{w_1, \dots, w_r\}.$ As before, the algorithm equates the coefficients of corresponding monomials on both sides of the equality and reduces it to QP.
\end{compactenum}

After the algorithm runs Step 3 as above, all standard constraints generated in Step 2 will be reduced to QP and can hence be passed to an external solver in Step 4, as illustrated in Section~\ref{sec:algo:overview}.

\paragraph{Degree bounds} We are using the same bound $d$ for the degree of the polynomials and templates in all cases above. This is not a requirement. One can fix different degree bounds for each case.

\paragraph{Handling Boundedness} To achieve completeness, we need the boundedness assumption, i.e.~that for every variable $v \in \vars,$ we always have $-M \leq v \leq M.$ To model this in the algorithm, we can add the boundedness inequalities to the left-hand side of every standard constraint. Additionally, we can get a concrete value for $M$ as part of the input, or treat $M$ symbolically, i.e.~as a template variable, and let the QP solver synthesize a value for it.

\subsection{Soundness and Completeness}\label{sec:algo_sound_complete}
\label{sec:algo:proof}

We now prove that our algorithm is sound and semi-complete for TBSP. The soundness result needs no extra assumptions and can be obtained directly. The completeness on the other hand relies on several assumptions: (i)~boundedness, (ii)~having chosen a large enough degree bound $d$, and (iii)~having invariants and post-conditions that, when written in DNF form, consist only of \emph{strongly positive} polynomial inequalities of the form $g \gg 0.$

\begin{theorem}[Soundness]
	Given a sketch polynomial program or equivalently a sketch polynomial transition system (SPTS) of the form $(\vars, \tvars, \locs,\loc_0,\theta_0,\loc_f, \theta_f, \transitions, \invariant)$ together with a cutset $\cutset$ as input, every concrete polynomial transition system (PTS) synthesized by the algorithm above is inductively valid.
\end{theorem}
\begin{proof}
The standard constraints of form~\eqref{eq:entailment-form} generated in Step 2 are equivalent to the initiation, consecution and finalization constraints in the definition of inductive validity. The reduction from standard constraints to QP in Step 3 is also sound since, in every case, it either writes the RHS of the standard constraint as a combination of the LHS polynomials, hence proving that it holds, or otherwise proves that the LHS is unsatisfiable and thus the standard constraint holds vacuously. 
\end{proof}

\begin{theorem}[Semi-completeness] \label{thm:complete}
	Consider a solvable sketch polynomial transition system (SPTS) of the form $(\vars, \tvars, \locs,\loc_0,\theta_0,\loc_f, \theta_f, \transitions, \invariant),$ together with a cutset $\cutset,$ that is given as input. Moreover, assume that:
	\begin{compactenum}
		\item The boundedness assumption holds, i.e.~there is a constant $M \in (0, \infty)$ such that for every $v \in \vars,$ we always have $-M \le v \le M.$
		\item Every invariant $\invariant(\loc)$ and post-condition $\theta_f,$ when written in disjunctive normal form, contains only strongly positive polynomial inequalities of the form $g \gg 0.$
	\end{compactenum}
	Then, there exists a constant degree bound $d \in \mathbb N$, for which the algorithm above is guaranteed to successfully synthesize an inductively valid polynomial transition system (PTS).
\end{theorem}
\begin{proof}
	Since our instance is solvable, there is a valuation $\tval$ for the template variables that yields an inductively valid PTS. We prove that for large enough $d$, the valuation $\tval$ is obtained by one of the solutions of the QP instance solved in Step 4 of the algorithm. The proof is pretty straightforward since we just have to check that all steps of our algorithm are complete. Step 2 is complete since it simply rewrites the inductive validity constraints as an equivalent set of standard constraints. For Step 3, we prove completeness of each case separately. Step 3.(i) is complete due to Farkas' Lemma (Corollary~\ref{thm:farkas_2}). In Step 3.(ii), if $\Phi$ is unsatisfiable, then the algorithm is complete based on Corollary~\ref{thm:farkas_2}. Otherwise, it is complete based on Corollary~\ref{thm:handel_2}. However, since we are using a degree bound $d$ for the monoid, the completeness only holds if the chosen $d$ is large enough. Moreover, Corollary~\ref{thm:handel_2} requires strong positivity of $g$, which corresponds to invariants and post-conditions in our use-case, and boundedness of $\sat(\Phi),$ which is a direct consequence of our boundedness assumption. Finally, Step 3.(iii) is complete due to Corollary~\ref{thm_pos_put_2} and Theorem~\ref{technical_tool_1}. These depend on $d$, strong positivity and boundedness in the exact same manner as in the case of Step 3.(ii).
\end{proof}

\paragraph{Limitations of Completeness} The main limitation in our completeness result is that it holds only if the degree bound $d$ chosen for the template polynomials is large enough. This is why we call it a \emph{semi-}completeness theorem. In theory, it is possible to come up with adversarial instances in which the required degree is exponentially high~\cite{mai2022complexity}. In practice, we rarely, if ever, need to use a higher degree than that of the polynomials that are already part of the input SPTS. The second limitation is boundedness. This limitation cannot be lifted since both Handelman's Theorem and Putinar's Positivstellensatz assume compactness, which is equivalent to being closed and bounded in $\mathbb R^n$. Nevertheless, it does not have a significant practical effect and the algorithm remains sound even without this assumption. It is also noteworthy that the treatment of the linear/affine case using Farkas' Lemma requires neither boundedness nor any specific value of $d$ and is always complete. Finally, strong positivity in the invariants and post-conditions is needed because Putinar's Positivstellensatz and Theorem~\ref{technical_tool_1} can only provide a sound and complete characterization of strongly positive polynomials over a bounded semi-algebraic set $X \subseteq R^n$ if we do not assume that $X$ itself is closed. In terms of the synthesis problem, this means that our algorithm is not guaranteed to be complete for inequalities of the form $f > 0$ in the invariants/post-conditions for which the value of $f$ in the runs of the program can get arbitrarily close to $0.$ However, this limitation is also not significant in practice because (i)~our soundness does not depend on it, and (ii)~$f + \epsilon \gg 0$ holds for any constant $\epsilon > 0.$ So, a small change (by any value $\epsilon > 0$) in the invariants/postconditions leads to an instance over which our completeness holds.

\section{Decidability and NP-hardness} \label{sec:complexity}

In this section, we study the decision variant of TBSP, in which we are given an SPTS of the form $\Gamma = (\vars, \tvars, \locs,\loc_0,\theta_0,\loc_f, \theta_f, \transitions, \invariant)$ and should decide whether there exists a valuation $\tval: \tvars \rightarrow \mathbb R$ for the template variables under which $\Gamma$ turns into an inductively valid PTS. 

\paragraph{Reduction to the First-order Theory of the Reals} Suppose $\vars = \{v_1, \dots, v_m\}$ and $\tvars = \{t_1, \dots, t_{m'}\}.$ Recall that Step 2 of our algorithm is sound and complete and reduces all the constraints to the form shown in Equation~\eqref{eq:firstorder}. However, in the synthesis algorithm, we are looking for a valuation to the template variables $\tvars$ that ensures~\eqref{eq:firstorder}. Hence, in the decision variant, we are interested in deciding the following first-order formula over the reals:
\begin{equation} \label{eq:tarski}
	\exists t_1, t_2, \ldots, t_{m'} ~\;~\; \forall v_1, v_2, \ldots, v_m ~\;~\;~ (\kappa_1 ~\wedge~ \kappa_2 ~\wedge~ \ldots ~\wedge~ \kappa_u)
\end{equation} 
It is well-known that the first-order theory of the reals is decidable and hence so is our problem. Specifically, one can apply Tarski's quantifier elimination method~\cite{tarski1949decision} to~\eqref{eq:tarski}. Note that while this proves decidability in theory, decision procedures and quantifier elimination on the first-order theory of the reals are notoriously unscalable and since~\eqref{eq:tarski} contains a quantifier alternation, it is practically beyond the reach of all modern solvers, even for toy programs.

\paragraph{NP-hardness} We prove our problem is strongly NP-hard, even for linear/affine transition systems, by providing a reduction from 3SAT. Consider a 3SAT instance 
$
\bigwedge_{i=1}^m \bigvee_{j=1}^3 l_{i, j}
$ over the boolean variables $\{x_1, \ldots, x_n\}, $ i.e.~each literal $l_{i, j}$ is either an $x_k$ or a $\neg x_k.$ We define the expression $e(l_{i,j})$ as $x_k$ if $l_{i, j} = x_k$ and $1-x_k$ if $l_{i, j} = \neg x_k.$ To reduce this 3SAT instance to our problem, we let $\vars = \{x_1, \ldots, x_n\}$ and $\tvars = \{t_1, \ldots, t_n\}.$ We then consider the sketch transition system corresponding to the following program:

\begin{lstlisting}[language=C,mathescape=true,numbers=none]
	@real: $x_1, x_2, \dots, x_n$
	@pre: $x_1 = t_1 ~\wedge~ x_2 = t_2 ~\wedge~ \dots ~\wedge~ x_n = t_n$;
	if ($x_1 > 0.5$)
		$x_1$ = $1$;
	else
		$x_1$ = $0$;
		$\vdots$
	if ($x_n > 0.5$)
		$x_n$ = $1$;
	else
		$x_n$ = $0$;   
	@post: $\bigwedge_{i=1}^m e(l_{i, 1})+e(l_{i, 2})+e(l_{i, 3}) > 0.5$;
\end{lstlisting}
Note that the template variables are appearing in the pre-condition. It is easy to see that the synthesis instance has a solution if and only if the 3SAT instance is satisfiable.

\section{Implementation and Experimental Results}\label{secion_proof_of_concept} \label{sec:exper}

In this section, we report on a prototype implementation of our algorithm and our experimental results.

\paragraph{Implementation} We implemented our approach in Python and used SymPy~\cite{meurer2017sympy} for symbolic computations. We also used the Z3 SMT solver~\cite{moura2008z3} to handle the final QP instances. In all cases, we used quadratic templates for the holes and set the technical parameter (degree upper-bound) $d$ to $2.$

\paragraph{Experimental Setting} All experimental results were obtained on an Intel i9-10980HK Processor (8 cores, 16 threads, 5.3 GHz, 16 MB Cache) with 32 GB of RAM running Microsoft Windows 10. 

\paragraph{Heuristics} We used the following simple heuristics to speed up the solution of the QP instance by the SMT solver:
\begin{compactitem}
	\item \emph{Simplifying Standard Constraints:} We apply the simplification procedure illustrated in Example~\ref{ex:step2}. Specifically, it is often the case that for some $v_1, v_2 \in \vars \cup \vars'$, the standard constraint contains $v_1 = v_2$ on its left-hand side. In these cases, we merge $v_1$ and $v_2$ into a single variable. Similarly, if the left-hand side of a standard constraint has $v_1 = c$ for a real constant $c$, then we simply replace every occurrence of $v_1$ with $c.$  
	\item \emph{Strengthening the Constraints:} Our algorithm introduces a large number of template variables, not only in Step 1, but also in Step 3, where each of the Stellens\"atze lead to the creation of new coefficient variables. Such new variables are created, for example, in the linear combinations appearing in Farkas/Handelman and the sum-of-squares template polynomials in Putinar. For most real-world cases, the vast majority of the template variables should be set to $0$ in solutions to the QP. However, their inclusion makes the QP much larger and more complicated and can even lead to the failure of other heuristics used by the SMT solver. Hence, we use a process which is intuitively the opposite of abstraction-refinement. We start with a strengthened version of our QP instance and repeatedly make it weaker, but never weaker than the original instance, until the SMT solver finds a solution. Specifically, we first add additional constraints of the form $t=0$ for every template variable $t$. This strengthens the constraints but makes the problem easier for the SMT solver, since it can now simply plug in the $0$ values and solve a lower-dimensional QP. Additionally, it would very likely lead to an unsatisfiable QP. However, Z3 can then provide us with an unsatisfiability \emph{core}, i.e.~a minimal subset of QP constraints that are unsatisfiable. If the core does not include any of our $t=0$ constraints, then we know that the original QP is also unsatisfiable. Otherwise, we simply remove all the $t=0$ constraints in the core and repeat the same process.
\end{compactitem}

\paragraph{Experimental Results} We now illustrate several programs that we synthesized using our implementation. In each case, we show the completed program, in which the synthesized parts are boxed. Due to space limitations, some examples are relegated to Appendix~\ref{app:ex}. Table~\ref{tab:runtime} provides a summary of the runtimes of our algorithm over these examples.

\begin{figure}
\begin{tabular}{c}
\begin{lstlisting}[language=C,numbers=none,mathescape=true]
	@prog: ClosestCubeRoot
	@real: $a,x,s,r$;
	@pre: $a \ge 0$;
	$x = a$;
	$r = 1$;
	s = 3.25;
	@invariant: $x \ge 0 \wedge -12 \cdot r^2 + 4 \cdot s = 1 \wedge 4 \cdot r^3 - 6 \cdot r^2 + 3 \cdot r + 4 \cdot x - 4 \cdot a = 1$ 
	while ($x-s \ge 0$)
	{
		$x = \Hole{x-s}$;
		$s = \Hole{3+s+6\cdot r}$;
		$r = r + 1$;
	}
	@post: $4 \cdot r^3 + 6 \cdot r^2 + 3 \cdot r \geq 4 \cdot a \geq 4 \cdot r^3 - 6 \cdot r^2 + 3 \cdot r - 1$
\end{lstlisting}

	\end{tabular}
\end{figure}

\begin{figure}
	\begin{tabular}{ccc}
		\begin{lstlisting}[language=C,numbers=none,mathescape=true]
@prog: ClosestSquareRoot
@real: $a,x,r$;
@pre: $a \ge 1$;
$x = 0.5 \cdot a$;
$r = 0$;
@invariant: $a = \Hole{2 \cdot x + r^2 - r} ~\wedge~ x \ge 0$
while ($x \ge r$) 
{
	$x = \Hole{x-r}$;
	$r = r + 1$;
}
@post: $r^2-r \ge a-2 \cdot r ~\wedge~ r^2 - r \le a$
		\end{lstlisting}
		&
		
\begin{lstlisting}[language=C,numbers=none,mathescape=true]
@prog: SquareRootFloor
@real: $a,su,t,n$;
@pre: $n \ge 0$;
$a = 0$;
$su = 1$;
$t = 1$;
@invariant: $a^2 \le n \wedge t = 2 \cdot a + 1 \wedge su = (a+1)^2$
while ($su \le n$)  
{
	$a = a + 1$;
	$t = \Hole{2+t}$;
	$su = \Hole{su + t}$;
}
@post((a * a) <= n);
\end{lstlisting}
	\end{tabular}
\end{figure}

\begin{figure}
	\begin{tabular}{cc}
		\begin{lstlisting}[language=C,numbers=none,mathescape=true]
@prog: SquareRootApproximation
@real: $a,err,r,q,p$;
@pre: $a \ge 1 \wedge err \ge 0$;
$r = a -1$;
$q = 1$;
$p = 0.5$;
@invariant: $p \ge 0 \wedge r \ge 0 \wedge a = q^2 + 2 \cdot r \cdot p$
while ($2 \cdot p \cdot r \ge err$) 
{
	if ($2 \cdot r - 2 \cdot q - p \ge 0$) 
	{
		$r = \Hole{2\cdot r -2\cdot q -p}$;
		$q = p + \Hole{q}$;
		$p = p / 2$;
	}
	else 
	{
		$r = \Hole{2 \cdot r}$;    
		$p = p \cdot 0.5$;
	}
}
@post:$q^2 \geq a - err \wedge q^2 \leq a$
\end{lstlisting}
&
\begin{lstlisting}[language=C,numbers=none,mathescape=true]
@prog: ConsecutiveCubes
@real: $N,n,x,y,z,s$;
@pre: $1 \ge 0$;
$n = 0$;
$x = 0$;
$y = 1$;
$z = 6$;
$s = 0$;
@invariant: $z = 6 \cdot n \wedge y = 3\cdot n^2+3\cdot n+1 \wedge x=n^3$
while ($n \le N$) 
{
	if ($x = n^3$) 
	{
		$s = s + x$;
	}
	$n = n + 1$;
	$x = \Hole{y+x}$;
	$y = \Hole{y+z}$;
	$z = \Hole{6+z}$;
}
@post: $1 \ge 0$;
\end{lstlisting}
			
	\end{tabular}
\end{figure}

\begin{compactitem}
    \item \textbf{PositivityEnforcement}: We are synthesizing an expression such that adding it to the given polynomial makes it positive for all possible inputs.

    \item \textbf{SquareCompletion}: We are synthesizing a number such that adding it to the given polynomial makes it always positive. 

    \item \textbf{SlidingBody}: 
    Consider the following Physics problem: An object is placed on an inclined plane of given length $l$. The inclination is given in terms of the sine and cosine of its angle. The coefficients of static and kinetic friction are also known. We want to find whether the body will slide and, if so, find the time it will take for it to reach the bottom of the plane.

    A physics student can easily write a program template as shown in the example for solving this problem using the principles and equations from classical mechanics. For example, the student can use an \textit{if} block to decide if the body will move or not. Then, inside the block, they can use an equation of motion that holds for the various parameters involved in the problem. Now, to find the time it takes to slide down, they can leave a program hole in the template. 
    
    \item \textbf{ArchimedesPrinciple}:
    A cuboid with dimensions ($l$, $b$, $h$) of density $\rho$ is placed in a fluid of density $\rho'$. We want to determine whether the body will float, and if it floats, find the percentage of the volume of the body that will be outside the fluid.
%
%    Similar to \ref{friction}, a program template can be easily written using Archimedes' Principle and equations from classical mechanics such that synthesizing the program will give us the answer to both the questions.
%
%    \item \ref{poly_approx} \textbf{PolynomialApproximation} In this problem, we are given two polynomials $f(x)$ and $g(x)$, a threshold $t$, universal bound $N$, and an error bound $\epsilon$, we want to synthesize a polynomial $h$ such that $h(x) - f(x) <= \epsilon$ if $x<t$, and $h(x) - g(x) <= \epsilon$ if $x>t$, for all $-N \le x \le N$. 

     \item \textbf{ClosestCubeRoot}\cite{rodriguez2018some}: Given a positive real number $n$, we need to synthesize a program that computes the closest integer cube root of $n$.

    \item \textbf{SquareRootFloor}\cite{rodriguez2018some}: Given a positive real number $n$, we need to synthesize a program that computes the floor of the square root of $n$.

    \item \textbf{SquareRootApproximation}\cite{rodriguez2018some}: Given a positive real number $n$ and an error threshold $\epsilon$, we need to synthesize a program that computes $\sqrt n$.

    \item \textbf{ConsecutiveCubes}\cite{rodriguez2018some}: We want to synthesize a program that computes the sum of first $n$ cubes.
\end{compactitem}

\begin{table}
	\begin{footnotesize}
	\begin{tabular}{l|l|l|l}
	Example & Step 2 & Step 3 & Step 4 (Z3) \\
	\hline
	ArchimedesPrinciple & 0.10 & 0.27 & 0.22 \\
	ClosestCubeRoot & 0.13 & 12.48 & 3.20 \\
	ClosestSquareRoot & 0.09 & 2.21 & 1.31 \\
	ConsecutiveCubes & 0.14 & 35.00 & 20.92 \\
	PolynomialApproximation & 0.09 & 0.10 & 0.05 \\
	PositivityEnforcement & 0.08 & 0.08 & 0.03 \\
	SlidingBody & 0.10 & 0.30 & 0.22 \\
	SquareCompletion & 0.07 & 0.08 & 0.04 \\
	SquareRootApproximation & 0.12 & 14.87 & 5.54 \\
	SquareRootFloor & 0.11 & 8.09 & 2.31
	\end{tabular}
	\end{footnotesize}
\caption{The runtimes of the different steps of our algorithm in seconds.}
\label{tab:runtime}
\end{table}

In summary, our algorithm is able to handle various polynomial synthesis problems, including examples from physics and ~\cite{rodriguez2018some} in a few seconds. Hence, not only is the algorithm semi-complete in theory, but it is also applicable in practice.

\section*{Acknowledgments}
This section was removed to preserve anonymity.

\newpage
\bibliography{refs}

\appendix

\newpage
\begin{landscape}
\section*{Comparison with Rosette and Sketch}
\begin{table}
	\begin{tabular}{|c|c|c|c|c|c|c|c|c|}
		\hline
		{\textbf{Benchmark}} & \multicolumn{2}{c|}{\textbf{Our Approach}} & \multicolumn{2}{c|}{\textbf{Sketch}} & \multicolumn{2}{c|}{\textbf{Rosette}} & \multicolumn{2}{c|}{\textbf{Direct}} \\ \cline{2-9} 
		& \multicolumn{1}{c|}{\textit{Result}} & \multicolumn{1}{c|}{\textit{Time}} & \multicolumn{1}{c|}{\textit{Result}} & \multicolumn{1}{c|}{\textit{Time}} & \multicolumn{1}{c|}{\textit{Result}} & \multicolumn{1}{c|}{\textit{Time}} & \multicolumn{1}{c|}{\textit{Result}} & \multicolumn{1}{c|}{\textit{Time}} \\ \hline
		ArchimedesPrinciple & \success & 0.59 & \unsat & 0.33 & \success & 0.37 & \success & 0.25 \\ \hline
		ClosestCubeRoot & \success & 15.81 & \wa & 0.46 & \tl & >43200 & \tl & >43200 \\ \hline
		ClosestSquareRoot & \success & 3.61 & \fail & 969.17 & \success & 0.15 & \tl & >43200\\ \hline
		ConsecutiveCubes & \success & 56.06 & \fail & 778.04 & \tl & >43200 & \tl & >43200 \\ \hline
		PolynomialApproximation & \success & 0.24 & \wa & 0.45 & \success & 0.12 & \success & 0.14 \\ \hline
		PositivityEnforcement & \success & 0.19 & \success & 0.30 & \tl & >43200 & \success & 0.09\\ \hline
		SlidingBody & \success & 0.62 & \fail & 0.47 & \fail & < 0.01 & \success & 0.13 \\ \hline
		SquareCompletion & \success & 0.19 & \wa & 0.29 & \success & 0.10 & \success & 0.09 \\ \hline
		SquareRootApproximation & \success &  20.53 & \wa & 0.39 & \tl & >43200 & \tl & >43200 \\ \hline
		Cohendiv & \success & 26.11 & \wa & 6.297 & \success & 0.11 & \success & 0.29 \\ \hline
		MannadivCarre & \success & 4.07 & \tl & >43200 & \success & 0.25 & \success & 316.5 \\ \hline
		MannadivCube & \success & 4.03 & \wa & 0.33 & \tl & >43200 & \success & 92.9 \\ \hline
		MannadivInd & \success & 5.25 & \fail & 0.408 & \success & 0.11 & \success & 0.18 \\ \hline
		Wensley & \success & 52.75 & \wa & 0.44 & \tl & >43200 & \tl & >43200 \\ \hline
		Euclidex2 & \success & 29.68 & \tl & >43200 & \success & 0.34 & \tl & >43200 \\ \hline
		LCM1 & \success & 10.99 & \tl & >43200 & \tl & >43200 & \success &  87.29\\ \hline
		Fermat1 & \success & 17.95 & \tl & >43200 & \success & 0.17 & \success & 117.14 \\ \hline
		Fermat2 & \success & 9.30 & \tl & >43200 & \tl & >43200 & \success & 0.31 \\ \hline
		Petter3 & \success & 2.51 & \fail & 0.493 & \tl & >43200 & \success & 124.51 \\ \hline
		Petter5 & \success & 1.84 & \fail & 2.72 & \tl & >43200 & \success & 117.33 \\ \hline
	\end{tabular}
\caption{Comparison of our approach with previous synthesis methods. All times are in seconds.}
\label{tab:comp}
\end{table}

Table~\ref{tab:comp} provides a comparison between our approach and the state-of-the-art synthesis tools Sketch and Rosette, as well as a direct encoding to SMT that was passed to Z3. The rest of this document contains the details of programs output by these tools. All experiments were run on the same machine as in our original submission with a time limit of 12 hours per instance. Note that Sketch is not sound for real variables and wrong answers are expected. However, in most cases Sketch's output is just setting everything to 0, which is wrong even for integer variables.

\end{landscape}

\newpage
\section{More Examples from Section~\ref{sec:exper}} \label{app:ex}
\begin{figure}[H]
	\begin{tabular}{ccc}
		\begin{lstlisting}[language=C,numbers=none,mathescape=true]
@prog: PositivityEnforcement
@real: $x,n$;
@pre: $x \le 100$;
$x$ = $n^2 -5 \cdot n$;
$n$ = $\Hole{5 \cdot n + 1}$;
$x$ = $x + n$;
@post: $x \geq 0$;    
		\end{lstlisting}
		
		&
		
		\begin{lstlisting}[language=C,numbers=none,mathescape=true]
@prog: SquareCompletion
@real: $x,n$;
@pre: $x \le 100$;
$x$ = $n^2 - 5\cdot n$;
$n$ = $\Hole{\frac{25}{4}}$;
$x$ = $x + n$;
@post: $x \geq 0$;    
		\end{lstlisting}
		
		&
		
		\begin{lstlisting}[language=C,numbers=none,mathescape=true]
@prog: PolynomialApproximation	
@real $t,x,e,g_1,g_2,n$;
@macro $f_1 = x^3$;
@macro $f_2 = x^3+x$;
@macro $f_3 = \Hole{5+x+x^3}$;
@pre: $1 \geq 0$;
$t = 0$;
$e = 5$;
if ($x \le t$) 
{
	$g_1 = f_3 - f_1$;
	$n = g_1$;
}
else 
{
	$g_2 = f_3 - f_2$;
	$n = g_2$;
}
@post: $n \le e$;
			
		\end{lstlisting}

	\end{tabular}
\end{figure} 

\begin{figure}[H]
	\begin{tabular}{ccc}
		\begin{lstlisting}[language=C,numbers=none,mathescape=true]
@prog: SlidingBody
@real: $move,e,t,g,nus,nuk,sin,cos,l$;
@pre: $t \geq 0 \wedge nus \geq nuk$;
$g = 9.8$;
$nuk = 0.2$;
$nus = 0.3$;
$sin = 0.5$;
$cos = 0.86$;
$l = 10$;
if ($sin \geq nus \cdot cos$) 
{
	$move = 1$;
	$t = \Hole{\frac{21}{2}-8.0056060865 \cdot move}$;
	$e = t \cdot t \cdot g \cdot sin - g \cdot nuk \cdot cos - 2 \cdot l$;
}
else 
{
	$move = 0$;
	$e = 0$;
	$t = \Hole{\frac{21}{2}-8.0056060865 \cdot move}$;
}
@post: $e = 0 \wedge t \geq 0;$
		\end{lstlisting}
		
		&
		
		\begin{lstlisting}[language=C,numbers=none,mathescape=true]
@prog: ArchimedesPrinciple
@real: $fl,m,l,b,h,\rho,e,x$;
@pre: $x \geq 0 \wedge m \geq 0$;
$l = 10$;
$b = 5$;
$h = 2$;
$\rho = 2$;
if ( $m \leq \rho \cdot l \cdot  b \cdot h$) 
{
	$fl = 1$;
	$x = \Hole{\frac{fl}{200}}$;
	$e = m - \rho \cdot l \cdot b \cdot h \cdot x$;
}
else 
{
	$e = 0$;
	$fl = 0$;
	$x = \Hole{\frac{fl}{200}}$;
}
@post: $e = 0 \wedge 0 \le x \le 1$;
		\end{lstlisting}
	\end{tabular}
\end{figure}

\newpage
\begin{figure}[H]
	\begin{tabular}{ccc}
		\begin{lstlisting}[language=C,numbers=none,mathescape=true]
@prog: Cohendiv
@real: $x,y,q,r,dd$;
@pre: $x \geq 0 \wedge y \geq 1$;
$ q = 0$;
$ r = x$;
@invariant: $x = q \cdot y + r \wedge r \geq 0$  
while ($r \geq y$)
{
	$d = 1$;
	$dd = y$;
	@invariant: $dd = y \cdot d \wedge x = q\cdot y + r $
	$\wedge r \geq 0 \wedge r \geq y \cdot d$
	while($r \geq 2 \cdot dd$)
	{
		$d = 2 \cdot d$;
		$dd = \Hole{2 \cdot dd}$;
	}
	$r = r - dd$;
	$q = \Hole{q + d}$;
}
@post: $1 \geq 0$;    
		\end{lstlisting}
		
% 		&
		
% 		\begin{lstlisting}[language=C,numbers=none,mathescape=true]
% @prog: SquareCompletion
% @real: $x,n$;
% @pre: $x \le 100$;
% $x$ = $n^2 - 5\cdot n$;
% $n$ = $\Hole{\frac{25}{4}}$;
% $x$ = $x + n$;
% @post: $x \geq 0$;    
% 		\end{lstlisting}
		
		&
		
		\begin{lstlisting}[language=C,numbers=none,mathescape=true]
@prog: MannadivCarre	
@real $n,y,x,t$;
@pre: $n \geq 0$;
$y = n$;
$x = 0$;
$t = 0$;
@invariant: $x^2 + 2\cdot t + y - n = 0$ 
while($y\cdot (y-1) \geq 0$)
{
	if($t = x$)
	{
		$y = \Hole{y - 1}$;
		$t = 0$;
		$x = x + 1$;
	}
	else
	{
		$y = y - 2$;
		$t = t + 1$;
	}
}
@post: $1 \geq 0$;
\end{lstlisting}

	\end{tabular}
\end{figure} 

\begin{figure}[H]
	\begin{tabular}{ccc}
		\begin{lstlisting}[language=C,numbers=none,mathescape=true]
@prog: MannadivCube	
@real $n,y,x,t$;
@pre: $n \geq 0$;
$y = n$;
$x = 0$;
$t = 0$;
@invariant: $x^3 + 3\cdot t + y - n = 0$ 
while($y \geq 0$)
{
	if($t = x*x$)
	{
		$y = \Hole{y - 3\cdot x - 1}$;
		$t = 0$;
		$x = x + 1$;
	}
	else
	{
		$y = y - 3$;
		$t = t + 1$;
	}
}
@post: $1 \geq 0$;
		\end{lstlisting}
		
		&
		
		\begin{lstlisting}[language=C,numbers=none,mathescape=true]
@prog: MannadivInd	
@real $x1,x2,y1,y2,y3$;
@pre: $x1 \geq 0 \wedge x2 \geq 0$;
$y1 = 0$;
$y2 = 0$;
$y3 = x1$;
@invariant: $y1\cdot x2  + y2 + y3 = x1$ 
while($y3 \geq 0$)
{
	if($y2 + 1 = x2$)
	{
		$y1 = y1 + 1$;
		$y2 = \Hole{0}$;
		$y3 = y3 - 1$;
	}
	else
	{
		$y2 = \Hole{y2 + 1}$;
		$y3 = y3 - 1$;
	}
}
@post: $1 \geq 0$;
		\end{lstlisting}
	\end{tabular}
\end{figure}

\newpage
\begin{figure}[H]
	\begin{tabular}{ccc}
		\begin{lstlisting}[language=C,numbers=none,mathescape=true]
@prog: Wensley
@real: $P,Q,E,a,b,d,y$;
@pre: $Q \geq P \wedge P \geq 0 \wedge E \geq 0$;
$ a = 0$;
$ b = Q/2$;
$ d = 1$;
$ y = 0$;
@invariant: $ a = Q\cdot y \wedge b = Q\cdot d/2$
$\wedge P - d\cdot Q \leq y \cdot Q \wedge y\cdot Q \leq P \wedge Q \geq 0$   
while ($d \geq E$)
{
	if($P \leq a + b$)
	{
		$b = \Hole{b/2}$;
		$d = d/2$;
	}
	else
	{
		$a = a + b$;
		$y = \Hole{y + d/2}$;
		$ b = b/2$;
		$ d = d/2$;
	}
}
@post: $P \geq y \cdot Q \wedge y\cdot Q \geq P - E \cdot Q$;    
		\end{lstlisting}
		
% 		&
		
% 		\begin{lstlisting}[language=C,numbers=none,mathescape=true]
% @prog: SquareCompletion
% @real: $x,n$;
% @pre: $x \le 100$;
% $x$ = $n^2 - 5\cdot n$;
% $n$ = $\Hole{\frac{25}{4}}$;
% $x$ = $x + n$;
% @post: $x \geq 0$;    
% 		\end{lstlisting}
		
		&
		
		\begin{lstlisting}[language=C,numbers=none,mathescape=true]
@prog: Euclidex2	
@real $a,b,p,q,r,s,x,y$;
@pre: $x \geq 0 \wedge y \geq 0$;
$a = x$;
$b = y$;
$p = 1$;
$q = 0$;
$r = 0$;
$s = 1$;
@invariant: $a = y \cdot r + x \cdot p \wedge b = x \cdot q + y \cdot s$
while($1 \geq 0$)
{
	if($a \geq b$)
	{
		$a = \Hole{a-b}$;
		$p = \Hole{p-q}$;
		$r = r - s$;
	}
	else
	{
		$b = \Hole{b-a}$;
		$q = q - p$;
		$s = \Hole{s-r}$;
	}
}
@post: $1 \geq 0$;
\end{lstlisting}

	\end{tabular}
\end{figure} 

\begin{figure}[H]
	\begin{tabular}{ccc}
		\begin{lstlisting}[language=C,numbers=none,mathescape=true]
@prog: LCM1	
@real $a,b,x,y,u,v$;
@pre: $a \geq 0 \wedge b \geq 0$;
$x = a$;
$y = b$;
$u = b$;
$v = 0$;
@invariant: $x\cdot u + y \cdot v = a\cdot b$ 
while($1 \geq 0$)
{
	@invariant: $x\cdot u + y \cdot v = a\cdot b$
	while($1 \geq 0$)
	{
		$x = \Hole{x-y}$;
		$v = v + u$;
	}
	@invariant: $x\cdot u + y \cdot v = a\cdot b$
	while($1 \geq 0$)
	{
		$y = y - x$;
		$u = \Hole{u+v}$;
	}
}
@post: $1 \geq 0$;
\end{lstlisting}
		
		&
		
		\begin{lstlisting}[language=C,numbers=none,mathescape=true]
@prog: Fermat1
@real $N,R,u,v,r$;
@pre: $N \geq 1 \wedge (R-1)^2  \leq N-1 \wedge N \leq R^2$;
$u = 2\cdot R + 1$;
$v = 1$;
$r = R^2 - N$;
@invariant: $N\geq 1 \wedge 4 \cdot (N + r) = u^2- v^2 - 2 \cdot u + 2 \cdot v$
while($1\geq 0$)
{
		@invariant: $N\geq 1 \wedge 4 \cdot (N + r) = u^2- v^2 - 2 \cdot u + 2 \cdot v$
		while($1 \geq 0$)
		{
			$r = r - v$;
			$v = \Hole{v + 2}$;
		}
		@invariant: $N\geq 1 \wedge 4 \cdot (N + r) = u^2- v^2 - 2 \cdot u + 2 \cdot v$
		while($1 \geq 0$)
		{
			$r = \Hole{r+u}$;
			$u = u + 2$;
		}
}
@post: $1 \geq 0$;
		\end{lstlisting}
	\end{tabular}
\end{figure}

\newpage
\begin{figure}[H]
	\begin{tabular}{ccc}
		\begin{lstlisting}[language=C,numbers=none,mathescape=true]
@prog: Fermat2 
@real $N,R,u,v,r$;
@pre: $N \geq 1 \wedge (R-1)^2  \leq N-1 \wedge N \leq R^2$;
$u = 2\cdot R + 1$;
$v = 1$;
$r = R^2 - N$;
@invariant: $N\geq 1 \wedge 4 \cdot (N + r) = u^2- v^2 - 2 \cdot u + 2 \cdot v$
while($r\geq 1$)
{
		if($r \geq 0$)
		{
			$r = r - v$;
			$v = v + 2$;
		}
		else
		{
			$r = \Hole{r+u}$;
			$u = u + 2$;
		}
}
@post: $1 \geq 0$; 
		\end{lstlisting}
		
% 		&
		
% 		\begin{lstlisting}[language=C,numbers=none,mathescape=true]
% @prog: SquareCompletion
% @real: $x,n$;
% @pre: $x \le 100$;
% $x$ = $n^2 - 5\cdot n$;
% $n$ = $\Hole{\frac{25}{4}}$;
% $x$ = $x + n$;
% @post: $x \geq 0$;    
% 		\end{lstlisting}
		
		&
		
		\begin{lstlisting}[language=C,numbers=none,mathescape=true]
@prog: Petter3	
@real $n,k,N,s$;
@pre: $N \geq 0$;
$s = 0$;
$n = 0$;
@invariant: $s = 0.25 \cdot n^2 \cdot (n+1)^2 \wedge n \leq N + 1$
while($n \leq N$)
{
	$s = s + \Hole{nˆ3}$;
	$n = n + 1$;	
}
@post: $1 \geq 0$;
\end{lstlisting}

	\end{tabular}
\end{figure} 

\begin{figure}[H]
	\begin{tabular}{ccc}
		\begin{lstlisting}[language=C,numbers=none,mathescape=true]
@prog: Petter5	
@real $n,k,N,s$;
@pre: $N \geq 0$;
$s = 0$;
$n = 0$;
@invariant: $s = n^6 \wedge n \leq N + 1$
while($n \leq N$)
{
	$s = s + \Hole{6 \cdot n^5 + 15 \cdot n^4 + 20 \cdot n^3 + 15 \cdot n^2 + 6 \cdot n + 1}$;
	$n = n + 1$;	
}
@post: $1 \geq 0$;
\end{lstlisting}
	\end{tabular}
\end{figure}

% \newpage

\newpage
\section*{Synthesis Results Using Sketch}

\begin{figure}[h]
\begin{tabular}{c}
\begin{lstlisting}[language=C,numbers=none,mathescape=true]
	@prog: ClosestCubeRoot
	@real: $a,x,s,r$;
	@pre: $a \ge 0$;
	$x = a$;
	$r = 1$;
	s = 3.25;
	@invariant: $x \ge 0 \wedge -12 \cdot r^2 + 4 \cdot s = 1 \wedge 4 \cdot r^3 - 6 \cdot r^2 + 3 \cdot r + 4 \cdot x - 4 \cdot a = 1$ 
	while ($x-s \ge 0$)
	{
		$x = \Hole{0}$;
		$s = \Hole{0}$;
		$r = r + 1$;
	}
	@post: $4 \cdot r^3 + 6 \cdot r^2 + 3 \cdot r \geq 4 \cdot a \geq 4 \cdot r^3 - 6 \cdot r^2 + 3 \cdot r - 1$
\end{lstlisting}

	\end{tabular}
\end{figure}

\begin{figure}
	\begin{tabular}{ccc}
		\begin{lstlisting}[language=C,numbers=none,mathescape=true]
@prog: ClosestSquareRoot
@real: $a,x,r$;
@pre: $a \ge 1$;
$x = 0.5 \cdot a$;
$r = 0$;
@invariant: $a = \Hole{\errorRejected} ~\wedge~ x \ge 0$
while ($x \ge r$) 
{
	$x = \Hole{\errorRejected}$;
	$r = r + 1$;
}
@post: $r^2-r \ge a-2 \cdot r ~\wedge~ r^2 - r \le a$
		\end{lstlisting}
		&
		
\begin{lstlisting}[language=C,numbers=none,mathescape=true]
@prog: SquareRootFloor
@real: $a,su,t,n$;
@pre: $n \ge 0$;
$a = 0$;
$su = 1$;
$t = 1$;
@invariant: $a^2 \le n \wedge t = 2 \cdot a + 1 \wedge su = (a+1)^2$
while ($su \le n$)  
{
	$a = a + 1$;
	$t = \Hole{12+8\cdot t}$;
	$su = \Hole{2\cdot su + 15 \cdot t +  8}$;
}
@post($a^2 \leq n$);
\end{lstlisting}
	\end{tabular}
\end{figure}

\begin{figure}
	\begin{tabular}{cc}
		\begin{lstlisting}[language=C,numbers=none,mathescape=true]
@prog: SquareRootApproximation
@real: $a,err,r,q,p$;
@pre: $a \ge 1 \wedge err \ge 0$;
$r = a -1$;
$q = 1$;
$p = 0.5$;
@invariant: $p \ge 0 \wedge r \ge 0 \wedge a = q^2 + 2 \cdot r \cdot p$
while ($2 \cdot p \cdot r \ge err$) 
{
	if ($2 \cdot r - 2 \cdot q - p \ge 0$) 
	{
		$r = \Hole{0}$;
		$q = p + \Hole{0}$;
		$p = p / 2$;
	}
	else 
	{
		$r = \Hole{0}$;    
		$p = p \cdot 0.5$;
	}
}
@post:$q^2 \geq a - err \wedge q^2 \leq a$
\end{lstlisting}
&
\begin{lstlisting}[language=C,numbers=none,mathescape=true]
@prog: ConsecutiveCubes
@real: $N,n,x,y,z,s$;
@pre: $1 \ge 0$;
$n = 0$;
$x = 0$;
$y = 1$;
$z = 6$;
$s = 0$;
@invariant: $z = 6 \cdot n \wedge y = 3\cdot n^2+3\cdot n+1 \wedge x=n^3$
while ($n \le N$) 
{
	if ($x = n^3$) 
	{
		$s = s + x$;
	}
	$n = n + 1$;
	$x = \Hole{\errorRejected}$;
	$y = \Hole{\errorRejected}$;
	$z = \Hole{\errorRejected}$;
}
@post: $1 \ge 0$;
\end{lstlisting}
			
	\end{tabular}
\end{figure}

\begin{figure}[H]
	\begin{tabular}{ccc}
		\begin{lstlisting}[language=C,numbers=none,mathescape=true]
@prog: PositivityEnforcement
@real: $x,n$;
@pre: $x \le 100$;
$x$ = $n^2 -5 \cdot n$;
$n$ = $\Hole{23 \cdot n + 23}$;
$x$ = $x + n$;
@post: $x \geq 0$;    
		\end{lstlisting}
		
		&
		
		\begin{lstlisting}[language=C,numbers=none,mathescape=true]
@prog: SquareCompletion
@real: $x,n$;
@pre: $x \le 100$;
$x$ = $n^2 - 5\cdot n$;
$n$ = $\Hole{6}$;
$x$ = $x + n$;
@post: $x \geq 0$;    
		\end{lstlisting}
		
		&
		
		\begin{lstlisting}[language=C,numbers=none,mathescape=true]
@prog: PolynomialApproximation	
@real $t,x,e,g_1,g_2,n$;
@macro $f_1 = x^3$;
@macro $f_2 = x^3+x$;
@macro $f_3 = \Hole{0}$;
@pre: $1 \geq 0$;
$t = 0$;
$e = 5$;
if ($x \le t$) 
{
	$g_1 = f_3 - f_1$;
	$n = g_1$;
}
else 
{
	$g_2 = f_3 - f_2$;
	$n = g_2$;
}
@post: $n \le e$;
			
		\end{lstlisting}

	\end{tabular}
\end{figure} 

\begin{figure}[H]
	\begin{tabular}{ccc}
		\begin{lstlisting}[language=C,numbers=none,mathescape=true]
@prog: SlidingBody
@real: $move,e,t,g,nus,nuk,sin,cos,l$;
@pre: $t \geq 0 \wedge nus \geq nuk$;
$g = 9.8$;
$nuk = 0.2$;
$nus = 0.3$;
$sin = 0.5$;
$cos = 0.86$;
$l = 10$;
if ($sin \geq nus \cdot cos$) 
{
	$move = 1$;
	$t = \Hole{\errorRejected}$;
	$e = t \cdot t \cdot g \cdot sin - g \cdot nuk \cdot cos - 2 \cdot l$;
}
else 
{
	$move = 0$;
	$e = 0$;
	$t = \Hole{\errorRejected}$;
}
@post: $e = 0 \wedge t \geq 0;$
		\end{lstlisting}
		
		&
		
		\begin{lstlisting}[language=C,numbers=none,mathescape=true]
@prog: ArchimedesPrinciple
@real: $fl,m,l,b,h,\rho,e,x$;
@pre: $x \geq 0 \wedge m \geq 0$;
$l = 10$;
$b = 5$;
$h = 2$;
$\rho = 2$;
if ( $m \leq \rho \cdot l \cdot  b \cdot h$) 
{
	$fl = 1$;
	$x = \Hole{\errorUnsat}$;
	$e = m - \rho \cdot l \cdot b \cdot h \cdot x$;
}
else 
{
	$e = 0$;
	$fl = 0$;
	$x = \Hole{\errorUnsat}$;
}
@post: $e = 0 \wedge 0 \le x \le 1$;
		\end{lstlisting}
	\end{tabular}
\end{figure}

\newpage 

\begin{figure}[H]
	\begin{tabular}{ccc}
		\begin{lstlisting}[language=C,numbers=none,mathescape=true]
@prog: Cohendiv
@real: $x,y,q,r,dd$;
@pre: $x \geq 0 \wedge y \geq 1$;
$ q = 0$;
$ r = x$;
@invariant: $x = q \cdot y + r \wedge r \geq 0$  
while ($r \geq y$)
{
	$d = 1$;
	$dd = y$;
	@invariant: $dd = y \cdot d \wedge x = q\cdot y + r $
	$\wedge r \geq 0 \wedge r \geq y \cdot d$
	while($r \geq 2 \cdot dd$)
	{
		$d = 2 \cdot d$;
		$dd = \Hole{27 \cdot dd + 16}$;
	}
	$r = r - dd$;
	$q = \Hole{10 \cdot + 13}$;
}
@post: $1 \geq 0$;    
		\end{lstlisting}
		
% 		&
		
% 		\begin{lstlisting}[language=C,numbers=none,mathescape=true]
% @prog: SquareCompletion
% @real: $x,n$;
% @pre: $x \le 100$;
% $x$ = $n^2 - 5\cdot n$;
% $n$ = $\Hole{\frac{25}{4}}$;
% $x$ = $x + n$;
% @post: $x \geq 0$;    
% 		\end{lstlisting}
		
		&
		
		\begin{lstlisting}[language=C,numbers=none,mathescape=true]
@prog: MannadivCarre	
@real $n,y,x,t$;
@pre: $n \geq 0$;
$y = n$;
$x = 0$;
$t = 0$;
@invariant: $x^2 + 2\cdot t + y - n = 0$ 
while($y\cdot (y-1) \geq 0$)
{
	if($t = x$)
	{
		$y = \Hole{\errorRejected}$;
		$t = 0$;
		$x = x + 1$;
	}
	else
	{
		$y = y - 2$;
		$t = t + 1$;
	}
}
@post: $1 \geq 0$;
\end{lstlisting}

	\end{tabular}
\end{figure} 

\begin{figure}[H]
	\begin{tabular}{ccc}
		\begin{lstlisting}[language=C,numbers=none,mathescape=true]
@prog: MannadivCube	
@real $n,y,x,t$;
@pre: $n \geq 0$;
$y = n$;
$x = 0$;
$t = 0$;
@invariant: $x^3 + 3\cdot t + y - n = 0$ 
while($y \geq 0$)
{
	if($t = x*x$)
	{
		$y = \Hole{0}$;
		$t = 0$;
		$x = x + 1$;
	}
	else
	{
		$y = y - 3$;
		$t = t + 1$;
	}
}
@post: $1 \geq 0$;
		\end{lstlisting}
		
		&
		
		\begin{lstlisting}[language=C,numbers=none,mathescape=true]
@prog: MannadivInd	
@real $x1,x2,y1,y2,y3$;
@pre: $x1 \geq 0 \wedge x2 \geq 0$;
$y1 = 0$;
$y2 = 0$;
$y3 = x1$;
@invariant: $y1\cdot x2  + y2 + y3 = x1$ 
while($y3 \geq 0$)
{
	if($y2 + 1 = x2$)
	{
		$y1 = y1 + 1$;
		$y2 = \Hole{\errorRejected}$;
		$y3 = y3 - 1$;
	}
	else
	{
		$y2 = \Hole{\errorRejected}$;
		$y3 = y3 - 1$;
	}
}
@post: $1 \geq 0$;
		\end{lstlisting}
	\end{tabular}
\end{figure}

\newpage
\begin{figure}[H]
	\begin{tabular}{ccc}
		\begin{lstlisting}[language=C,numbers=none,mathescape=true]
@prog: Wensley
@real: $P,Q,E,a,b,d,y$;
@pre: $Q \geq P \wedge P \geq 0 \wedge E \geq 0$;
$ a = 0$;
$ b = Q/2$;
$ d = 1$;
$ y = 0$;
@invariant: $ a = Q\cdot y \wedge b = Q\cdot d/2$
$\wedge P - d\cdot Q \leq y \cdot Q \wedge y\cdot Q \leq P \wedge Q \geq 0$   
while ($d \geq E$)
{
	if($P \leq a + b$)
	{
		$b = \Hole{\errorRejected}$;
		$d = d/2$;
	}
	else
	{
		$a = a + b$;
		$y = \Hole{\errorRejected}$;
		$ b = b/2$;
		$ d = d/2$;
	}
}
@post: $P \geq y \cdot Q \wedge y\cdot Q \geq P - E \cdot Q$;    
		\end{lstlisting}
		
% 		&
		
% 		\begin{lstlisting}[language=C,numbers=none,mathescape=true]
% @prog: SquareCompletion
% @real: $x,n$;
% @pre: $x \le 100$;
% $x$ = $n^2 - 5\cdot n$;
% $n$ = $\Hole{\frac{25}{4}}$;
% $x$ = $x + n$;
% @post: $x \geq 0$;    
% 		\end{lstlisting}
		
		&
		
		\begin{lstlisting}[language=C,numbers=none,mathescape=true]
@prog: Euclidex2	
@real $a,b,p,q,r,s,x,y$;
@pre: $x \geq 0 \wedge y \geq 0$;
$a = x$;
$b = y$;
$p = 1$;
$q = 0$;
$r = 0$;
$s = 1$;
@invariant: $a = y \cdot r + x \cdot p \wedge b = x \cdot q + y \cdot s$
while($1 \geq 0$)
{
	if($a \geq b$)
	{
		$a = \Hole{\errorRejected}$;
		$p = \Hole{\errorRejected}$;
		$r = r - s$;
	}
	else
	{
		$b = \Hole{\errorRejected}$;
		$q = q - p$;
		$s = \Hole{\errorRejected}$;
	}
}
@post: $1 \geq 0$;
\end{lstlisting}

	\end{tabular}
\end{figure} 

\begin{figure}[H]
	\begin{tabular}{ccc}
		\begin{lstlisting}[language=C,numbers=none,mathescape=true]
@prog: LCM1	
@real $a,b,x,y,u,v$;
@pre: $a \geq 0 \wedge b \geq 0$;
$x = a$;
$y = b$;
$u = b$;
$v = 0$;
@invariant: $x\cdot u + y \cdot v = a\cdot b$ 
while($1 \geq 0$)
{
	@invariant: $x\cdot u + y \cdot v = a\cdot b$
	while($1 \geq 0$)
	{
		$x = \Hole{\errorRejected}$;
		$v = v + u$;
	}
	@invariant: $x\cdot u + y \cdot v = a\cdot b$
	while($1 \geq 0$)
	{
		$y = y - x$;
		$u = \Hole{\errorRejected}$;
	}
}
@post: $1 \geq 0$;
\end{lstlisting}
		
		&
		
		\begin{lstlisting}[language=C,numbers=none,mathescape=true]
@prog: Fermat1
@real $N,R,u,v,r$;
@pre: $N \geq 1 \wedge (R-1)^2  \leq N-1 \wedge N \leq R^2$;
$u = 2\cdot R + 1$;
$v = 1$;
$r = R^2 - N$;
@invariant: $N\geq 1 \wedge 4 \cdot (N + r) = u^2- v^2 - 2 \cdot u + 2 \cdot v$
while($1\geq 0$)
{
		@invariant: $N\geq 1 \wedge 4 \cdot (N + r) = u^2- v^2 - 2 \cdot u + 2 \cdot v$
		while($1 \geq 0$)
		{
			$r = r - v$;
			$v = \Hole{\errorRejected}$;
		}
		@invariant: $N\geq 1 \wedge 4 \cdot (N + r) = u^2- v^2 - 2 \cdot u + 2 \cdot v$
		while($1 \geq 0$)
		{
			$r = \Hole{\errorRejected}$;
			$u = u + 2$;
		}
}
@post: $1 \geq 0$;
		\end{lstlisting}
	\end{tabular}
\end{figure}

\newpage
\begin{figure}[H]
	\begin{tabular}{ccc}
		\begin{lstlisting}[language=C,numbers=none,mathescape=true]
@prog: Fermat2
@real $N,R,u,v,r$;
@pre: $N \geq 1 \wedge (R-1)^2  \leq N-1 \wedge N \leq R^2$;
$u = 2\cdot R + 1$;
$v = 1$;
$r = R^2 - N$;
@invariant: $N\geq 1 \wedge 4 \cdot (N + r) = u^2- v^2 - 2 \cdot u + 2 \cdot v$
while($r\geq 1$)
{
		if($r \geq 0$)
		{
			$r = r - v$;
			$v = v + 2$;
		}
		else
		{
			$r = \Hole{\errorRejected}$;
			$u = u + 2$;
		}
}
@post: $1 \geq 0$; 
		\end{lstlisting}
		
% 		&
		
% 		\begin{lstlisting}[language=C,numbers=none,mathescape=true]
% @prog: SquareCompletion
% @real: $x,n$;
% @pre: $x \le 100$;
% $x$ = $n^2 - 5\cdot n$;
% $n$ = $\Hole{\frac{25}{4}}$;
% $x$ = $x + n$;
% @post: $x \geq 0$;    
% 		\end{lstlisting}
		
		&
		
		\begin{lstlisting}[language=C,numbers=none,mathescape=true]
@prog: Petter3	
@real $n,k,N,s$;
@pre: $N \geq 0$;
$s = 0$;
$n = 0$;
@invariant: $s = 0.25 \cdot n^2 \cdot (n+1)^2 \wedge n \leq N + 1$
while($n \leq N$)
{
	$s = s + \Hole{ç}$;
	$n = n + 1$;	
}
@post: $1 \geq 0$;
\end{lstlisting}

	\end{tabular}
\end{figure} 

\begin{figure}[H]
	\begin{tabular}{ccc}
		\begin{lstlisting}[language=C,numbers=none,mathescape=true]
@prog: Petter5	
@real $n,k,N,s$;
@pre: $N \geq 0$;
$s = 0$;
$n = 0$;
@invariant: $s = n^6 \wedge n \leq N + 1$
while($n \leq N$)
{
	$s = s + \Hole{\errorRejected}$;
	$n = n + 1$;	
}
@post: $1 \geq 0$;
\end{lstlisting}
	\end{tabular}
\end{figure}

% \newpage 

\newpage
\section*{Synthesis Results Using Rosette}

\begin{figure}[h]
\begin{tabular}{c}
\begin{lstlisting}[language=C,numbers=none,mathescape=true]
	@prog: ClosestCubeRoot
	@real: $a,x,s,r$;
	@pre: $a \ge 0$;
	$x = a$;
	$r = 1$;
	s = 3.25;
	@invariant: $x \ge 0 \wedge -12 \cdot r^2 + 4 \cdot s = 1 \wedge 4 \cdot r^3 - 6 \cdot r^2 + 3 \cdot r + 4 \cdot x - 4 \cdot a = 1$ 
	while ($x-s \ge 0$)
	{
		$x = \Hole{\errorTimeout}$;
		$s = \Hole{\errorTimeout}$;
		$r = r + 1$;
	}
	@post: $4 \cdot r^3 + 6 \cdot r^2 + 3 \cdot r \geq 4 \cdot a \geq 4 \cdot r^3 - 6 \cdot r^2 + 3 \cdot r - 1$
\end{lstlisting}

	\end{tabular}
\end{figure}

\begin{figure}
	\begin{tabular}{ccc}
		\begin{lstlisting}[language=C,numbers=none,mathescape=true]
@prog: ClosestSquareRoot
@real: $a,x,r$;
@pre: $a \ge 1$;
$x = 0.5 \cdot a$;
$r = 0$;
@invariant: $a = \Hole{2 \cdot x + r^2 - r} ~\wedge~ x \ge 0$
while ($x \ge r$) 
{
	$x = \Hole{x-r}$;
	$r = r + 1$;
}
@post: $r^2-r \ge a-2 \cdot r ~\wedge~ r^2 - r \le a$
		\end{lstlisting}
		&
		
\begin{lstlisting}[language=C,numbers=none,mathescape=true]
@prog: SquareRootFloor
@real: $a,su,t,n$;
@pre: $n \ge 0$;
$a = 0$;
$su = 1$;
$t = 1$;
@invariant: $a^2 \le n \wedge t = 2 \cdot a + 1 \wedge su = (a+1)^2$
while ($su \le n$)  
{
	$a = a + 1$;
	$t = \Hole{2+t}$;
	$su = \Hole{su + t}$;
}
@post((a * a) <= n);
\end{lstlisting}
	\end{tabular}
\end{figure}

\begin{figure}
	\begin{tabular}{cc}
		\begin{lstlisting}[language=C,numbers=none,mathescape=true]
@prog: SquareRootApproximation
@real: $a,err,r,q,p$;
@pre: $a \ge 1 \wedge err \ge 0$;
$r = a -1$;
$q = 1$;
$p = 0.5$;
@invariant: $p \ge 0 \wedge r \ge 0 \wedge a = q^2 + 2 \cdot r \cdot p$
while ($2 \cdot p \cdot r \ge err$) 
{
	if ($2 \cdot r - 2 \cdot q - p \ge 0$) 
	{
		$r = \Hole{\errorTimeout}$;
		$q = p + \Hole{\errorTimeout}$;
		$p = p / 2$;
	}
	else 
	{
		$r = \Hole{\errorTimeout}$;    
		$p = p \cdot 0.5$;
	}
}
@post:$q^2 \geq a - err \wedge q^2 \leq a$
\end{lstlisting}
&
\begin{lstlisting}[language=C,numbers=none,mathescape=true]
@prog: ConsecutiveCubes
@real: $N,n,x,y,z,s$;
@pre: $1 \ge 0$;
$n = 0$;
$x = 0$;
$y = 1$;
$z = 6$;
$s = 0$;
@invariant: $z = 6 \cdot n \wedge y = 3\cdot n^2+3\cdot n+1 \wedge x=n^3$
while ($n \le N$) 
{
	if ($x = n^3$) 
	{
		$s = s + x$;
	}
	$n = n + 1$;
	$x = \Hole{\errorTimeout}$;
	$y = \Hole{\errorTimeout}$;
	$z = \Hole{\errorTimeout}$;
}
@post: $1 \ge 0$;
\end{lstlisting}
			
	\end{tabular}
\end{figure}

\begin{figure}[H]
	\begin{tabular}{ccc}
		\begin{lstlisting}[language=C,numbers=none,mathescape=true]
@prog: PositivityEnforcement
@real: $x,n$;
@pre: $x \le 100$;
$x$ = $n^2 -5 \cdot n$;
$n$ = $\Hole{\errorTimeout}$;
$x$ = $x + n$;
@post: $x \geq 0$;    
		\end{lstlisting}
		
		&
		
		\begin{lstlisting}[language=C,numbers=none,mathescape=true]
@prog: SquareCompletion
@real: $x,n$;
@pre: $x \le 100$;
$x$ = $n^2 - 5\cdot n$;
$n$ = $\Hole{6.25}$;
$x$ = $x + n$;
@post: $x \geq 0$;    
		\end{lstlisting}
		
		&
		
		\begin{lstlisting}[language=C,numbers=none,mathescape=true]
@prog: PolynomialApproximation	
@real $t,x,e,g_1,g_2,n$;
@macro $f_1 = x^3$;
@macro $f_2 = x^3+x$;
@macro $f_3 = \Hole{x^3 - 0.5\cdot x^2 - 0.5 \cdot x + 0.5}$;
@pre: $1 \geq 0$;
$t = 0$;
$e = 5$;
if ($x \le t$) 
{
	$g_1 = f_3 - f_1$;
	$n = g_1$;
}
else 
{
	$g_2 = f_3 - f_2$;
	$n = g_2$;
}
@post: $n \le e$;
			
		\end{lstlisting}

	\end{tabular}
\end{figure} 

\begin{figure}[H]
	\begin{tabular}{ccc}
		\begin{lstlisting}[language=C,numbers=none,mathescape=true]
@prog: SlidingBody
@real: $move,e,t,g,nus,nuk,sin,cos,l$;
@pre: $t \geq 0 \wedge nus \geq nuk$;
$g = 9.8$;
$nuk = 0.2$;
$nus = 0.3$;
$sin = 0.5$;
$cos = 0.86$;
$l = 10$;
if ($sin \geq nus \cdot cos$) 
{
	$move = 1$;
	$t = \Hole{\errorRose}$;
	$e = t \cdot t \cdot g \cdot sin - g \cdot nuk \cdot cos - 2 \cdot l$;
}
else 
{
	$move = 0$;
	$e = 0$;
	$t = \Hole{\errorRose}$;
}
@post: $e = 0 \wedge t \geq 0;$
		\end{lstlisting}
		
		&
		
		\begin{lstlisting}[language=C,numbers=none,mathescape=true]
@prog: ArchimedesPrinciple
@real: $fl,m,l,b,h,\rho,e,x$;
@pre: $x \geq 0 \wedge m \geq 0$;
$l = 10$;
$b = 5$;
$h = 2$;
$\rho = 2$;
if ( $m \leq \rho \cdot l \cdot  b \cdot h$) 
{
	$fl = 1$;
	$x = \Hole{\frac{fl}{200}}$;
	$e = m - \rho \cdot l \cdot b \cdot h \cdot x$;
}
else 
{
	$e = 0$;
	$fl = 0$;
	$x = \Hole{\frac{fl}{200}}$;
}
@post: $e = 0 \wedge 0 \le x \le 1$;
		\end{lstlisting}
	\end{tabular}
\end{figure}

\newpage
\begin{figure}[H]
	\begin{tabular}{ccc}
		\begin{lstlisting}[language=C,numbers=none,mathescape=true]
@prog: Cohendiv
@real: $x,y,q,r,dd$;
@pre: $x \geq 0 \wedge y \geq 1$;
$ q = 0$;
$ r = x$;
@invariant: $x = q \cdot y + r \wedge r \geq 0$  
while ($r \geq y$)
{
	$d = 1$;
	$dd = y$;
	@invariant: $dd = y \cdot d \wedge x = q\cdot y + r $
	$\wedge r \geq 0 \wedge r \geq y \cdot d$
	while($r \geq 2 \cdot dd$)
	{
		$d = 2 \cdot d$;
		$dd = \Hole{2 \cdot dd}$;
	}
	$r = r - dd$;
	$q = \Hole{q + d}$;
}
@post: $1 \geq 0$;    
		\end{lstlisting}
		
% 		&
		
% 		\begin{lstlisting}[language=C,numbers=none,mathescape=true]
% @prog: SquareCompletion
% @real: $x,n$;
% @pre: $x \le 100$;
% $x$ = $n^2 - 5\cdot n$;
% $n$ = $\Hole{\frac{25}{4}}$;
% $x$ = $x + n$;
% @post: $x \geq 0$;    
% 		\end{lstlisting}
		
		&
		
		\begin{lstlisting}[language=C,numbers=none,mathescape=true]
@prog: MannadivCarre	
@real $n,y,x,t$;
@pre: $n \geq 0$;
$y = n$;
$x = 0$;
$t = 0$;
@invariant: $x^2 + 2\cdot t + y - n = 0$ 
while($y\cdot (y-1) \geq 0$)
{
	if($t = x$)
	{
		$y = \Hole{y - 1}$;
		$t = 0$;
		$x = x + 1$;
	}
	else
	{
		$y = y - 2$;
		$t = t + 1$;
	}
}
@post: $1 \geq 0$;
\end{lstlisting}

	\end{tabular}
\end{figure} 

\begin{figure}[H]
	\begin{tabular}{ccc}
		\begin{lstlisting}[language=C,numbers=none,mathescape=true]
@prog: MannadivCube	
@real $n,y,x,t$;
@pre: $n \geq 0$;
$y = n$;
$x = 0$;
$t = 0$;
@invariant: $x^3 + 3\cdot t + y - n = 0$ 
while($y \geq 0$)
{
	if($t = x*x$)
	{
		$y = \Hole{\errorRejected}$;
		$t = 0$;
		$x = x + 1$;
	}
	else
	{
		$y = y - 3$;
		$t = t + 1$;
	}
}
@post: $1 \geq 0$;
		\end{lstlisting}
		
		&
		
		\begin{lstlisting}[language=C,numbers=none,mathescape=true]
@prog: MannadivInd	
@real $x1,x2,y1,y2,y3$;
@pre: $x1 \geq 0 \wedge x2 \geq 0$;
$y1 = 0$;
$y2 = 0$;
$y3 = x1$;
@invariant: $y1\cdot x2  + y2 + y3 = x1$ 
while($y3 \geq 0$)
{
	if($y2 + 1 = x2$)
	{
		$y1 = y1 + 1$;
		$y2 = \Hole{0}$;
		$y3 = y3 - 1$;
	}
	else
	{
		$y2 = \Hole{y2 + 1}$;
		$y3 = y3 - 1$;
	}
}
@post: $1 \geq 0$;
		\end{lstlisting}
	\end{tabular}
\end{figure}

\newpage
\begin{figure}[H]
	\begin{tabular}{ccc}
		\begin{lstlisting}[language=C,numbers=none,mathescape=true]
@prog: Wensley
@real: $P,Q,E,a,b,d,y$;
@pre: $Q \geq P \wedge P \geq 0 \wedge E \geq 0$;
$ a = 0$;
$ b = Q/2$;
$ d = 1$;
$ y = 0$;
@invariant: $ a = Q\cdot y \wedge b = Q\cdot d/2$
$\wedge P - d\cdot Q \leq y \cdot Q \wedge y\cdot Q \leq P \wedge Q \geq 0$   
while ($d \geq E$)
{
	if($P \leq a + b$)
	{
		$b = \Hole{\errorRejected}$;
		$d = d/2$;
	}
	else
	{
		$a = a + b$;
		$y = \Hole{\errorRejected}$;
		$ b = b/2$;
		$ d = d/2$;
	}
}
@post: $P \geq y \cdot Q \wedge y\cdot Q \geq P - E \cdot Q$;    
		\end{lstlisting}
		
% 		&
		
% 		\begin{lstlisting}[language=C,numbers=none,mathescape=true]
% @prog: SquareCompletion
% @real: $x,n$;
% @pre: $x \le 100$;
% $x$ = $n^2 - 5\cdot n$;
% $n$ = $\Hole{\frac{25}{4}}$;
% $x$ = $x + n$;
% @post: $x \geq 0$;    
% 		\end{lstlisting}
		
		&
		
		\begin{lstlisting}[language=C,numbers=none,mathescape=true]
@prog: Euclidex2	
@real $a,b,p,q,r,s,x,y$;
@pre: $x \geq 0 \wedge y \geq 0$;
$a = x$;
$b = y$;
$p = 1$;
$q = 0$;
$r = 0$;
$s = 1$;
@invariant: $a = y \cdot r + x \cdot p \wedge b = x \cdot q + y \cdot s$
while($1 \geq 0$)
{
	if($a \geq b$)
	{
		$a = \Hole{a-b}$;
		$p = \Hole{p-q}$;
		$r = r - s$;
	}
	else
	{
		$b = \Hole{b-a}$;
		$q = q - p$;
		$s = \Hole{s-r}$;
	}
}
@post: $1 \geq 0$;
\end{lstlisting}

	\end{tabular}
\end{figure} 

\begin{figure}[H]
	\begin{tabular}{ccc}
		\begin{lstlisting}[language=C,numbers=none,mathescape=true]
@prog: LCM1	
@real $a,b,x,y,u,v$;
@pre: $a \geq 0 \wedge b \geq 0$;
$x = a$;
$y = b$;
$u = b$;
$v = 0$;
@invariant: $x\cdot u + y \cdot v = a\cdot b$ 
while($1 \geq 0$)
{
	@invariant: $x\cdot u + y \cdot v = a\cdot b$
	while($1 \geq 0$)
	{
		$x = \Hole{\errorRejected}$;
		$v = v + u$;
	}
	@invariant: $x\cdot u + y \cdot v = a\cdot b$
	while($1 \geq 0$)
	{
		$y = y - x$;
		$u = \Hole{\errorRejected}$;
	}
}
@post: $1 \geq 0$;
\end{lstlisting}
		
		&
		
		\begin{lstlisting}[language=C,numbers=none,mathescape=true]
@prog: Fermat1
@real $N,R,u,v,r$;
@pre: $N \geq 1 \wedge (R-1)^2  \leq N-1 \wedge N \leq R^2$;
$u = 2\cdot R + 1$;
$v = 1$;
$r = R^2 - N$;
@invariant: $N\geq 1 \wedge 4 \cdot (N + r) = u^2- v^2 - 2 \cdot u + 2 \cdot v$
while($1\geq 0$)
{
		@invariant: $N\geq 1 \wedge 4 \cdot (N + r) = u^2- v^2 - 2 \cdot u + 2 \cdot v$
		while($1 \geq 0$)
		{
			$r = r - v$;
			$v = \Hole{v + 2}$;
		}
		@invariant: $N\geq 1 \wedge 4 \cdot (N + r) = u^2- v^2 - 2 \cdot u + 2 \cdot v$
		while($1 \geq 0$)
		{
			$r = \Hole{r+u}$;
			$u = u + 2$;
		}
}
@post: $1 \geq 0$;
		\end{lstlisting}
	\end{tabular}
\end{figure}

\newpage
\begin{figure}[H]
	\begin{tabular}{ccc}
		\begin{lstlisting}[language=C,numbers=none,mathescape=true]
@prog: Fermat2 
@real $N,R,u,v,r$;
@pre: $N \geq 1 \wedge (R-1)^2  \leq N-1 \wedge N \leq R^2$;
$u = 2\cdot R + 1$;
$v = 1$;
$r = R^2 - N$;
@invariant: $N\geq 1 \wedge 4 \cdot (N + r) = u^2- v^2 - 2 \cdot u + 2 \cdot v$
while($r\geq 1$)
{
		if($r \geq 0$)
		{
			$r = r - v$;
			$v = v + 2$;
		}
		else
		{
			$r = \Hole{\errorRejected}$;
			$u = u + 2$;
		}
}
@post: $1 \geq 0$; 
		\end{lstlisting}
		
% 		&
		
% 		\begin{lstlisting}[language=C,numbers=none,mathescape=true]
% @prog: SquareCompletion
% @real: $x,n$;
% @pre: $x \le 100$;
% $x$ = $n^2 - 5\cdot n$;
% $n$ = $\Hole{\frac{25}{4}}$;
% $x$ = $x + n$;
% @post: $x \geq 0$;    
% 		\end{lstlisting}
		
		&
		
		\begin{lstlisting}[language=C,numbers=none,mathescape=true]
@prog: Petter3	
@real $n,k,N,s$;
@pre: $N \geq 0$;
$s = 0$;
$n = 0$;
@invariant: $s = 0.25 \cdot n^2 \cdot (n+1)^2 \wedge n \leq N + 1$
while($n \leq N$)
{
	$s = s + \Hole{\errorRejected}$;
	$n = n + 1$;	
}
@post: $1 \geq 0$;
\end{lstlisting}

	\end{tabular}
\end{figure} 

\begin{figure}[H]
	\begin{tabular}{ccc}
		\begin{lstlisting}[language=C,numbers=none,mathescape=true]
@prog: Petter5	
@real $n,k,N,s$;
@pre: $N \geq 0$;
$s = 0$;
$n = 0$;
@invariant: $s = n^6 \wedge n \leq N + 1$
while($n \leq N$)
{
	$s = s + \Hole{\errorRejected}$;
	$n = n + 1$;	
}
@post: $1 \geq 0$;
\end{lstlisting}
	\end{tabular}
\end{figure}

%\newpage
%% Appendix
%\appendix
%\input{chapters/appendix}

\end{document}